\documentclass[11pt]{article} 

\usepackage[a4paper,left=2.5cm,right=2.5cm,top=2.5cm,bottom=3.0cm]{geometry}
\usepackage[dvips]{epsfig}
\usepackage{amsfonts}
\usepackage{moreverb}

\usepackage{amsthm}

\newtheorem{theorem}{Theorem}[section]
\newtheorem{lemma}{Lemma}[section]

\newtheorem{proposition}{Proposition}[section]

\newtheorem{observation}{Observation}[section]

\newcommand{\R}{{\cal{R}}}

\usepackage{graphicx}
\usepackage{xcolor}
\usepackage{enumitem}
\usepackage{amssymb}
\usepackage{amsmath}

\usepackage[noend, noline, ruled, linesnumbered]{algorithm2e} 
\SetAlFnt{\normalsize}
\DontPrintSemicolon
\SetKwComment{Comment}{$\triangleright$\ }{}
\SetKwProg{Def}{def}{:}{}
\SetArgSty{textnormal}
\SetKwRepeat{Do}{do}{while}
\SetKwComment{Comment}{$\triangleright$\ }{}

\newcommand{\marker}{M}
\newcommand{\D}{D}
\newcommand{\C}{\mathcal{C}}
\newcommand{\W}{\mathcal{W}}

\newcommand{\F}{\mathcal{F}}
\newcommand{\reversal}[1]{\varphi(#1)}
\newcommand{\distPC}{7}
\newcommand{\distPM}{3}
\newcommand{\distPCmin}{5}
\newcommand{\distMC}{4}

\begin{document}

\title{\bf Building a Nest by an Automaton}

\author{
Jurek Czyzowicz\footnotemark[1]\\
{\small D\'{e}partement d'informatique}\\
{\small Universit\'{e} du Qu\'{e}bec en Outaouais, Canada}\\{\small \sl jurek@uqo.ca}
\and
Dariusz Dereniowski\\
{\small Faculty of Electronics, Telecommunications}\\{\small and Informatics}\\ {\small Gdansk University of Technology, Poland}\\ {\small \sl deren@eti.pg.edu.pl}
\and
Andrzej Pelc\footnotemark[2]\\ 
{\small D\'{e}partement d'informatique}\\
{\small Universit\'{e} du Qu\'{e}bec en Outaouais, Canada}\\ {\small \sl pelc@uqo.ca}
}
\date{ }
\maketitle
\def\thefootnote{\fnsymbol{footnote}}

\footnotetext[1]{Supported in part by NSERC discovery grant }

\footnotetext[2]{Supported in part by NSERC discovery grant 2018-03899 and by
the Research Chair in Distributed Computing of the
Universit\'{e} du Qu\'{e}bec en Outaouais.}

\centerline{\today}


\maketitle

\begin{abstract}
A robot modeled as a deterministic finite automaton has to build a structure from material available to it. The robot navigates in the infinite oriented grid $\mathbb{Z} \times \mathbb{Z}$. Some cells of the grid are full (contain a brick) and others are empty. The subgraph of the grid induced by full cells, called the {\em field}, is initially connected. The (Manhattan) distance between the farthest cells of the field is called its {\em span}. The robot starts at a full cell. It can carry at most one brick at a time. At each step it can pick a  brick from a full cell, move to an adjacent cell and drop a brick at an empty cell. The aim of the robot is to construct the most compact possible structure composed of all bricks, i.e., a {\em  nest}. That is, the robot has to move all bricks in such a way that the span of the resulting field be the smallest.   

Our main result is the design of a deterministic finite automaton that accomplishes this task and subsequently stops,  for every initially connected field, in time $O(sz)$, where $s$ is the span of the initial field and $z$ is the number of bricks. We show that this complexity is optimal.
  \vspace{1ex}

\end{abstract}

\thispagestyle{empty}
\newpage
\setcounter{page}{1}

\baselineskip    0.19in
\parskip         0.0in
\parindent       0.3in

\section{Introduction}

{\bf The problem.}
A mobile agent (robot) modeled as a deterministic finite automaton has to build a structure from material available to it. The robot navigates in the infinite oriented grid $\mathbb{Z} \times \mathbb{Z}$ represented as the set of unit square cells in the two-dimensional plane,
with all cell sides vertical or horizontal. The robot has a compass enabling it to move from a currently occupied cell to one of the four cells (to the North, East, South, West) adjacent to it. Some cells of the grid contain a brick, i.e., 
 are {\em full}, other cells are {\em empty}. The subgraph of the grid induced by full cells, called the {\em field}, is initially connected. The (Manhattan) distance between the farthest cells of the field is called its {\em span}. Notice that the span of any current field may be much smaller than its diameter as a subgraph of  the grid. In fact, this diameter may be sometimes undefined, if the field becomes disconnected.
 The robot starts at a full cell. It can carry at most one brick at a time. At each step, the robot can pick a brick from the currently occupied full cell (if it does not carry any brick at this time), moves to an adjacent cell, and can drop a brick at the currently occupied empty cell (if it carries a brick). The robot has no a priori knowledge of the initial field, of its span or of the number of bricks.  
 
 The aim of the robot is to construct the most compact possible structure composed of all bricks, i.e., a {\em  nest}. That is, the robot has to move all bricks in such a way that the span of the resulting field be the smallest. The above task has many real applications. In the natural world, animals use material scattered in a territory (pieces of wood, small branches, leaves) to build a nest, and minimizing the span is important to better protect it. A mobile robot may be used to clean a territory littered by hazardous material, in which case minimizing the span of the resulting placement  of contaminated pieces facilitates subsequent decontamination. A more mundane example is the everyday task of sweeping the floor, whose aim is to gather all trash in a small space and then get rid of it. 

\noindent 
{\bf Our results.}
 Our main result is the design of a deterministic finite automaton that accomplishes the task of building a nest and subsequently stops,  for every initially connected field, in time $O(sz)$, where $s$ is the span of the initial field and $z$ is the number of bricks. The time is defined as the number of moves of the robot. We show that this complexity is optimal.
 
 The essence of our nest building algorithm is to instruct the robot to make a series of trips to get consecutive bricks, one at a time, and carry them to some designated compact area. This approach ensures the optimal complexity. (In order to show where the problem is, we also sketch a much simpler algorithm that uses another approach but has significantly larger complexity).
 There are two major difficulties to carry out this plan. The first is that the span of the initial field may be much larger than the memory of the robot, and hence the robot that already put several bricks in a compact area and goes for the next brick cannot remember the way back to the area where  it started building.
 Thus we need to prepare the way, so that the robot can recognize the backtrack path locally at each decision point. The second problem is that, while we temporarily disconnect the field during the execution of the algorithm, special care has to be taken so that the connected components of intermediary fields be close to each other, to prevent the robot from getting lost in large empty spaces.
 
 To the best of our knowledge, the task of constructing structures from available material using an automaton, has never been studied before in the algorithmic setting. It is interesting to compare this task to that of exploration of mazes by automata, that is a classic topic with over 50 years of history (see the section ``Related work''). It follows from the result of Budach \cite{B78} (translated to our terminology) that if an automaton can only navigate in the field without moving bricks then it cannot explore all connected fields, even without the stop requirement, i.e., it cannot even see all bricks. By contrast, it follows from our result that the ability of moving bricks enables the automaton not only to see all bricks but to build a potentially useful structure using all of them and stop, and to accomplish all of that with optimal complexity.

\noindent
{\bf The model.}
 We consider the infinite oriented grid $\mathbb{Z} \times \mathbb{Z}$ represented as the set of unit square cells tiling the two-dimensional plane,
with all cell sides vertical or horizontal. Each cell has 4 adjacent cells, North, East, South and West of it. Some cells of the grid contain a brick, i.e., 
 are {\em full}, other cells are {\em empty}. The subgraph of the grid induced by full cells is initially connected. 
 At each step of the algorithm this subgraph can change, due to the actions of the robot, described below. At each step, the subgraph induced by the full cells is called the current {\em field}.  Any maximal connected  subgraph of the current field is called a {\em component}.
 Throughout the paper, the {\em distance} between two cells $(x,y)$ and $(x',y')$  of the grid is the Manhattan distance between them, i.e., $|x-x'|+|y-y'|$.
 The number of cells of a field is called its {\em size}, and the distance between two farthest cells of a field is called its {\em span}. 
 A nest of size $z$ is a field that has the minimum span among all fields of size $z$.

We are given a mobile entity (robot) starting in some cell of the initial field and traveling in the grid. 
The robot has a priori no knowledge of the field, of its size or of its span.
The robot is formalized as a finite deterministic Mealy automaton
$\R=(X,Y,\S,\delta,\lambda,S_0, S_{\textup{f}})$. 
 $X=\{e,f\} \times \{l,h\}$ is the input alphabet, 
$Y=\{N,E,S,W\} \times \{e,f\} \times \{l,h\}$ is the output alphabet. $\S$ is a finite set of states
with two special states: $S_0$ called initial and $S_{\textup{f}}$ called final.
$\delta:\S\times
X\to\S$ is the state transition function, and $\lambda:\S\times X \to
Y$ is the output function.

The meaning of the input and output symbols is the following. At each step of its functioning, the robot is at some cell of the grid
and has some weight: it is either light, denoted by $l$ (does not carry a brick) or heavy, denoted by $h$ (carries a brick). Moreover, the current cell is either empty, denoted by $e$ or full, denoted by $f$.
The input $x\in \{e,f\} \times \{l,h\}$ gives the automaton information about these facts. The robot is in some state $S$. Given this state and the input $x$,
the robot outputs the symbol $\lambda(x,S)\in \{N,E,S,W\} \times \{e,f\} \times \{l,h\}$ with the following meaning. The first term indicates the adjacent cell to which the robot moves: North, East, South or West of the current cell. The second term determines whether the robot leaves the current cell empty or  full, and the third term indicates whether the robot transits as heavy or as light to the adjacent cell. Since the robot can only either leave the current cell intact and not change its own weight, or pick a brick from a full cell leaving it empty (in case when the robot was previously light), or drop a brick on an empty cell leaving it full (in case when the robot was previously heavy), we have the following restrictions on the possible values of the output function $\lambda$:

$\lambda(S,e,l)$ must be $(\cdot, e,l)$

$\lambda(S,e,h)$ must be either $(\cdot,e,h)$ or $(\cdot,f,l)$

 $\lambda(S,f,l)$ must be either $(\cdot,f,l)$ or $(\cdot,e,h)$
 
 $\lambda(S,f,h)$ must be $(\cdot,f,h)$
 
 Seeing the input symbol $x$ and being in a current state $S$, the robot makes the changes indicated by the output function (it goes to the indicated
 adjacent cell, possibly changes the filling of the current cell as indicated and possibly changes its own weight as indicated), and transits to state $\delta(x,S)$.
 The robot starts light in a full cell in the initial state $S_0$ (hence its initial input symbol is $(f,l)$) and terminates its action in the final state $S_{\textup{f}}$.

\noindent
{\bf Related work.}
Problems concerning exploration and navigation performed by mobile agents or  robots in an unknown
environment have been studied for many years
(cf.~\cite{H89,RKSI}). The relevant literature can be divided into two parts, according to the environment where the robots operate:
it can be either a geometric terrain, possibly with obstacles, or a network modeled as a graph in which the robot moves along edges.

In the geometric context, a closely related problem is that of pattern formation \cite{DFSY,DPV,SY}. Robots, modeled as points freely moving in the plane have to arrange themselves to form a pattern given as input. This task has been mostly studied in the context of  asynchronous oblivious robots having full visibility of other robots positions.

The graph setting can be further specified in two different
ways. In \cite{AH,BFRSV,BS,DP,FI04} the robot explores strongly
connected directed graphs and it can move only in the tail-to-head
direction of an edge, not vice-versa. In
\cite{ABRS,BRS2,B78,DFKP,DJMW,DKK,PaPe,R79-80} the explored
graph is undirected and the robot can traverse edges in both
directions. Graph exploration scenarios can be also classified in another important way. 
It is either assumed that nodes of the
graph have unique labels which the robot can recognize (as in,
e.g.,~\cite{DP,DKK,PaPe}), or it is assumed that nodes are anonymous
(as in, e.g.,~\cite{BFRSV,BS,B78,CDK,R79-80}). In our case, we work with the infinite anonymous grid, hence it is
an undirected anonymous graph scenario.
The efficiency measure adopted in papers dealing with
graph exploration is either the completion time of this task, measured
by the number of edge traversals, (cf., e.g.,~\cite{PaPe}), or the
memory size of the robot, measured either in bits or by the number of
states of the finite automaton modeling the robot (cf.,
e.g., \cite{DFKP,FIPPP,FI04}). We are not concerned with minimizing the memory size but we assume
that this memory is finite. However we want to minimize the time of our construction task.

The capability of a robot to explore anonymous undirected graphs has
been studied in, e.g., \cite{BK78,B78,DFKP,FIPPP,K79,R79-80}. In
particular, it was shown in \cite{R79-80} that no finite automaton can
explore all cubic planar graphs (in fact no finite set of finite
automata can cooperatively perform this task). Budach \cite{B78} proved that a single automaton cannot explore all mazes
(that we call connected fields in this paper).
Hoffmann \cite{H81} proved that one pebble does not help to do it. By contrast, Blum and Kozen \cite{BK78} showed that this task can be accomplished by two cooperating automata or by a single automaton with two pebbles.
The size of port-labeled graphs which cannot be explored by a given robot was
investigated in \cite{FIPPP}.

Recently a lot of attention has been devoted to the problem of searching for a target hidden in the infinite anonymous oriented grid by cooperating agents modeled as either deterministic or probabilistic automata. Such agents are sometimes called ants. It was shown in \cite{ELSUW} that 3 randomized or 4 deterministic automata can accomplish this task. Then matching lower bounds were proved: the authors of \cite{CELU} showed that 2 randomized automata are not enough for target searching in the grid, and the authors of \cite{BUW} proved that 3 deterministic automata are not enough for this task. Searching for a target in the infinite grid with obstacles was considered in \cite{LKUW}.

Our present paper adopts the same model of environment as the above papers, i.e., the infinite anonymous oriented grid. However the task we study is different: instead of searching for a target, the robot has to build some structure from the available material. To the best of our knowledge, such construction tasks performed by automata have never been studied previously in the algorithmic setting.

\section{Terminology and preliminaries}

In the description and analysis of our algorithm we will categorize full cells.
A full cell is said to be a \emph{border cell} if it is adjacent to an empty cell.
A full cell that has only one full adjacent cell is called a \emph{leaf}.
A full cell is called {\em special}, if it is either a leaf, or has at least two full cells adjacent to it, sharing a corner.

A finite deterministic automaton may remember a constant number of bits by encoding them in its states.
We will use this fact to define several simple procedures and simplifications that we use in the sequel.
The first simplification is as follows.
We formulate the actions of the robot based on the configuration of bricks in its neighborhood.
More precisely, at any step, the robot knows whether each cell at distance  at most $r=8$ from its current cell is full or empty.
This can be achieved by performing a bounded local exploration with return, after each move of the robot.

We will use the notion of current \emph{orientation} of the robot.
At the beginning of its navigation, the robot goes in one of the four cardinal directions. Then its orientation is determined in one of the two ways:
either by its last step (North, East, South or West) or by a \emph{turn}: we say that the robot \emph{turns left} (respectively \emph{right}) meaning that it changes its orientation in the appropriate way while remaining at the same cell. Clearly the robot can remember its orientation, using its states.
We refer to cardinal directions with respect to this current orientation. Thus, e.g.,  if the robot is oriented East then we say that its adjacent North
(resp. East, South or West) cell is 
\emph{left} (resp. \emph{in front}, \emph{right}, \emph{back}) of it.

Whenever we say that the robot located at a cell $c$  and not carrying a brick \emph{brings} a brick from a full cell $c'$ to $c$ we mean that 
the robot moves from  $c$ to $c'$, picks the brick from $c'$, moves back to $c$ and restores its original orientation.
Whenever we say that the robot located at a cell $c$ and carrying a brick \emph{places} it at an empty cell $c'$
we mean that  the robot moves from  $c$ to $c'$, drops the brick at $c'$, moves back to $c$ and restores its original orientation.

Whenever the robot selects a cell according to some condition that  is fulfilled by more than one cell, the robot selects the cell that is minimal with respect to the following total order $\prec$ on the set of all cells.
For  cells $c=(x,y)$ and $c'=(x',y')$, $c\prec c'$ holds if and only if either $y<y'$, or $y=y'$ and $x\leq x'$.
We denote by $|S|$ the number of cells in a sequence or a set $S$ of cells.

We define a {\em disc} of radius $r\geq 0$ with center $c$ to be the set of all cells at distance at most $r$ from $c$, see Fig. \ref{fig:disc}.  A disc of radius $r$ contains 
$z_r=2(1+3+5+\cdots +(2r-1))+(2r+1)=2r^2+2r+1$ cells and has span $2r$. A {\em rough disc} of size $z$, where $z_r\leq z<z_{r+1}$ is defined as follows. If $z=z_r$, then the rough disc is the disc of radius $r$. Otherwise, let $F$ be the set of cells not belonging to the disc $D$ of radius $r$ but adjacent to some cell of it. Add to $D$ exactly
$z-z_r$ cells belonging to $F$, starting from the North neighbor of the East-most cell of $D$ and going counterclockwise.
If $z_r<z \leq z_r+2r+2$ then the rough disc of size $z$ has span $2r+1$ and if $z_r+2r+2<z<z_{r+1}$ then the rough disc of size $z$ has span $2r+2$, 
the same as the disc of radius $r+1$ that has size $z_{r+1}$. 

\begin{figure}[htb!]
\begin{center}
\includegraphics[scale=0.5]{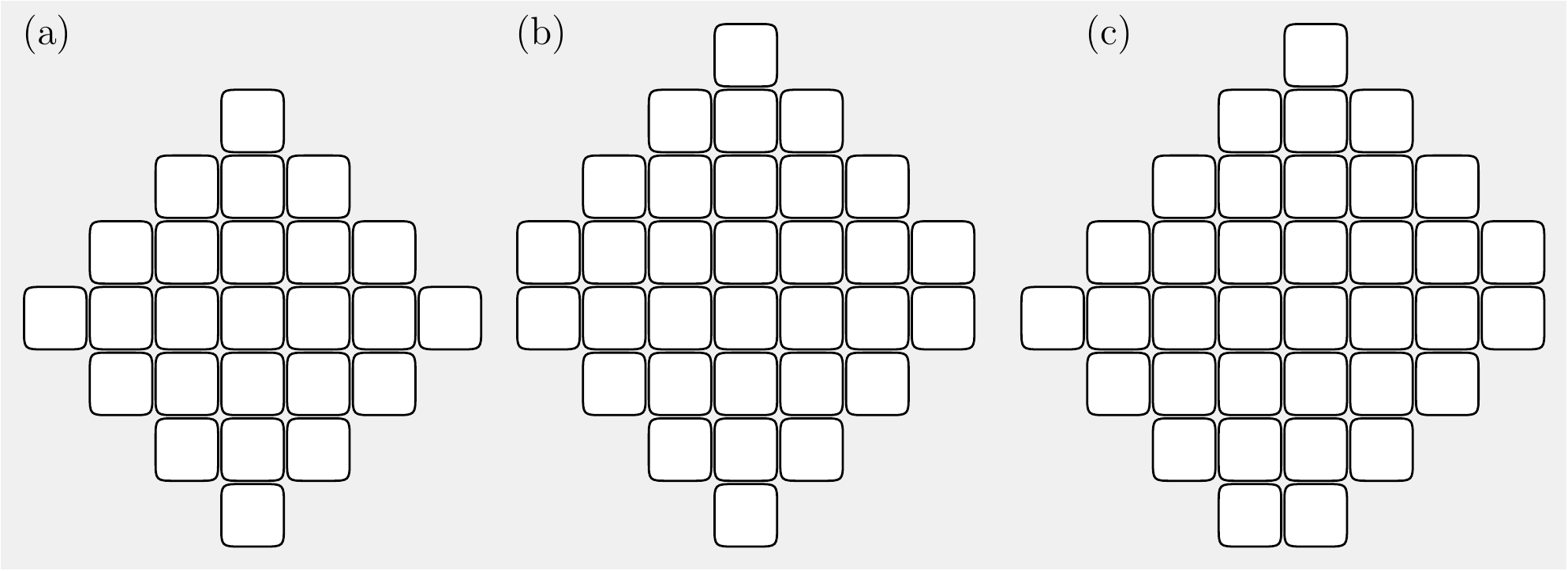}
\caption{(a) disc of size $z_r$ for $r=3$; (b) a rough disc of size
$z_r+7$ and span $2r+1$, $r=3$; (c) a rough disc of size $z_r+11$ and
span $2r+2$, $r=3$}
\label{fig:disc}
\end{center}
\end{figure}

\begin{proposition} \label{pro:nest}
Any rough disc is a nest.
\end{proposition}
\begin{proof}
We write $d(a,b)$ to denote the distance between two cells $a$ and $b$.
For a given span $s$, consider a field $C_s$ that has span $s$ and a maximum number of cells among fields of this span. 
Consider three consecutive full cells $f_1,f_2$ and $f_3$ of $C_s$ in a row or in a column.
Let $c_1,c_2$ and $c_3$ be three consecutive cells such that $c_i$ is adjacent to $f_i$, $i\in\{1,2,3\}$.
We prove that $c_2$ is full.
We proceed by contradiction, i.e., suppose that $c_2$ is empty.
Consider the field $C'$ obtained from $C_s$ by adding a brick to the cell $c_2$.
By definition of $C_s$, the span of $C'$ is larger than $s$.
Thus, there exists a shortest path $\W=(w_1,\ldots,w_{s+2})$ in the grid, that connects two full cells $w_1$ and $w_{s+2}$ of $C'$.
If $c_2$ is different from $w_1$ and $w_{s+2}$  then $C_s$ has span larger than $s$, which is a contradiction.
Thus,  without loss of generality, we can assume $c_2=w_{s+2}$.
We have $d(w_1,f_2)\geq s$ because the length of $\W$ is $s+1$ and $f_2$ and $c_2$ are neighbors.
Thus, since $f_2$ is a neighbor of both $f_1$ and $f_3$, we have either $d(f_1,w_{s+2})\geq s+1$ or $d(f_3,w_{s+2})\geq s+1$.
Since $f_1,f_3$ and $w_{s+2}$ are in $C_s$, this gives a contradiction with the span of $C_s$ being at most $s$. Hence $c_2$ is full.

This implies that $C_s$ has the following property: each border cell of $C_s$ is adjacent to two or three empty cells.
Define an {\em extremity} to be an East-most, West-most, North-most or South-most cell of a field.
Thus, either there exists a single East-most (respectively West-most, North-most or South-most) extremity, or there are two adjacent such extremities.
Moreover, each border cell of $C_s$ that is not an extremity does not have an adjacent border cell. (Intuitively, the border is composed of ``staircases'').
Denote by $\F_s$ the class of fields with the above two properties. Note that $C_s$ which has maximum size among fields of span $s$ belongs to $\F_s$.

We will compute the number of cells in any field $C$ in the class $\F_s$.
Partition the cells of $C$ into rows $R_1,\ldots,R_t$ counted from North to South.
We have proved that $|R_1|\in\{1,2\}$ and $|R_l|\in\{1,2\}$ because these sets consist of extremities only.
We first consider the case when there is a total of four extremities in $C$.
Thus, in particular, $|R_1|=|R_t|=1$. (Note that, in the considered case, $s$ must be even).
Let $j$ be the maximal index such that $|R_j|<|R_{j+1}|$ and let $j'$ be the minimal index such that $|R_{j'-1}|\leq|R_{j'}|$.
Since for each $i\in\{2,\ldots,j\}$, $|R_i|=|R_{i-1}|+2$, we obtain $\sum_{i=1}^{j}|R_i|=j^2$.
By symmetry, $\sum_{i=j'}^{t}|R_i|=j^2$.
Denote $m=j'-j-1$.
The span of $C$ equals $s=t+m-2=2j+2m-2$, and hence $m=s/2-j+1$.
We have $|R_i|=2j+1$ for each $i\in\{j+1,\ldots,j'-1\}$.
This implies $\sum_{i=j+1}^{j'-1}|R_i|=m(2j+1)=j(s+1)-2j^2+s/2+1$.
Thus, the size of $C$ is $z(j)=\sum_{i=1}^{t}|R_i|=j(s+1)+s/2+1$.
It follows that for a field of maximum size in $\F_s$ that has a total of four extremities, the index $j$ must be as large as possible, i.e., $j=s/2$.
Hence such a field is a disc of span $s$, i.e., of radius $r=s/2$.

Now consider all other fields in the class $\F_s$, i.e., those in which some of the extremities are pairs of adjacent cells. Every such field $F'$ can be obtained from a field $F$ in $\F_{s-1} \cup \F_{s-2}$ with exactly four extremities, by adding bricks in some cells adjacent to the border of $F$.
The number of cells in each such field $F'$ equals $|F|+c_1x+c_2y$ for some nonnegative constants $c_1$ and $c_2$,
where $x$ is the span of $F$ and $y$ is the index $j$ in the preceding paragraph, for the field $F$.
Consider a field $F'$ of maximum size among all fields in $\F_s$. Hence the corresponding field $F$ has maximum size among all fields in $\F_{s-1}$ or in $\F_{s-2}$
that have exactly four extremities. As proved above, such a field $F$ must be a disc.
Denote by $h(s)$ the maximum size of a field in $\F_s$.
In particular we have that $h(s)=\max\{z_{s/2},z_{(s-2)/2}+2s-1\}=z_{s/2}$ for even $s$, and $h(s)=z_{(s-1)/2}+s$ for odd $s$. 

Let $z$ be such that $z_r\leq z<z_{r+1}$ and consider a nest $N$ of size $z$. Denote its span by $s_z$.
We have three cases depending on the possible values of $z$.

First suppose that $z=z_r$.
We have $s_z\geq 2r$ because otherwise, for a maximum-size field $M$ in $\F_{2r-1}$, we would have $|M|\geq|N|$, and $|M|= h(2r-1)=z_{r-1}+2r=2r^2+1<z_r$, which would imply $|N|<z$, contradicting the definition of $N$.
Thus, the disc of size $z=z_r$ has span at most $s_z$ and thus it is a nest.

Next suppose that $z_r<z\leq z_r+2r+2$.
We have $s_z\geq 2r+1$ because otherwise, for a maximum-size field $M$ in $\F_{2r}$, we would have $|M|\geq|N|$, and $|M|= h(2r)=z_{r}<z$, which would imply $|N|<z$, contradicting the definition of $N$.
Hence, the rough disc of size $z$ has span at most $s_z$ and thus it is a nest.

Finally suppose that $z_r+2r+2<z< z_{r+1}$.
We have $s_z\geq 2r+2$ because otherwise, for a maximum-size field $M$ in $\F_{2r+1}$, we would have $|M|\geq|N|$, and $|M|\leq h(2r+1)=z_{r}+2r+1<z$, which would imply $|N|<z$, contradicting the definition of $N$.
This proves that the rough disc of size $z$ has span at most $s_z$ and thus it is a nest.
\end{proof}

The nests built by our automaton will be discs or rough discs, depending on the number of available bricks.
The following proposition shows that the complexity of our nest-building algorithm is optimal, regardless of the relation between the size of the initial field and its span (recall that, by definition, the span must be smaller than the size $z$ and  it must be in $\Omega(\sqrt{z})$). Our lower bound on complexity follows from geometric properties of the grid, and hence it holds regardless of the machine that builds the nest,
i.e.,  even if the robot is a Turing machine knowing a priori the initial field.

\begin{proposition}\label{pro:lb}
Let $s'<z$ be positive integers such that $s' \in \Omega(\sqrt{z})$. There exists an initial field of size $z$ and span $s\in \Theta(s')$, such that any algorithm that builds a nest starting from this field must use time $\Omega(sz)$. 
\end{proposition}

\begin{proof}
The proposition is proved by considering an appropriate initial field that requires large time to be transformed to a nest.
We define a {\em rough rectangle} of size $z$ and sides $a \leq b$, such that $z-ab<b$, to be the following set of cells: there is a grid $a \times b$ of cells
(with $a$ rows and $b$ columns), and the remaining $z-ab$ cells are attached to the North-most cells of the $z-ab$ West-most columns. 

\noindent
Case 1. $s' \geq 10\sqrt{z}$.

Take the rough rectangle $R$ with sides $b=s'$, $a=\lfloor z/b\rfloor$ as the initial field. The span of $R$ is $s\in \Theta(s')$.
There exist subsets $A$ and $B$ of $R$ of size at least $z/3$ at distance at least  $s'/4$: $A$ is composed of the $\lceil  z/(3a)\rceil$ West-most columns of the grid and $B$ is composed of the $\lceil  z/(3a)\rceil$ East-most columns of the grid.
Since the span of a nest of size $z$
is smaller than $2\sqrt{z} \leq s'/5$, in order to transform $R$ into a nest, at least $z/3$ bricks have to be moved at distance at least $s'/40$. Hence the time of building a nest from the initial field $R$ is at least $s'z/120 \in \Omega (sz)$.

\noindent
Case 2. $s' < 10\sqrt{z}$.

In this case we have $s'\in \Theta(\sqrt{z})$.
Take the rough rectangle $R$ with sides $b=\lceil10\sqrt{z}\rceil$, $a=\lfloor z/b\rfloor$ as the initial field. The span of $R$ is $s\in  \Theta(\sqrt{z})= \Theta(s')$
and $s \geq 10\sqrt{z}$.
There exist subsets $A$ and $B$ of $R$ of size at least $z/3$ at distance at least  $s/4$. Since the span of a nest of size $z$
is smaller than $2\sqrt{z} \leq s/5$, in order to transform $R$ into a nest, at least $z/3$ bricks have to be moved at distance at least $s/40$. Hence the time of building a nest from the initial field $R$ is at least $sz/120 \in \Omega (sz)$.
\end{proof}

\section{The algorithm}

The robot will move bricks from the original field and build two special components.
One of them will be a rough disc that will be gradually extended.
The second one will be a one-cell component whose only cell is called the \emph{marker}. The robot will periodically get at large distances from the rough disc being built, and the role of the marker will be to indicate to the robot that it got back in the vicinity of the rough disc.
Any component that is different from the   rough disc and from the marker will be called a \emph{free component}.

During the execution of the entire algorithm, the robot will not ensure that the full cells outside of the   rough disc and of the marker form one component --- they may form several components --- but after adding a new brick to the rough disc these components will be always at a bounded distance, i.e., at distance $O(1)$, from the   rough disc that the robot is constructing.

We are now ready to sketch the high-level idea of the algorithm, whose pseudo-code is presented at the end of this section as algorithm \ref{alg:nest}.
First the robot performs some preliminary actions by establishing the marker and the initial   rough disc and by calling procedure \ref{alg:sweeping}, which together result in constructing the first   rough disc $\D$ (of size one), placing the marker next to it and ensuring that no full cells other than the marker are at distance at most $\distPC$ from $\D$.
Then each iteration of the main loop of algorithm \ref{alg:nest} performs three actions.
First, it executes procedure \ref{alg:find-next-brick} that allows the robot to find a brick in a free component $\C$.  
This brick must be carefully chosen. For example, greedily picking the closest available brick would soon result in creating large empty spaces between  components of the field, in which the robot could get lost.
This brick will be later used to extend the   rough disc.
However, this procedure may lead the robot far from the   rough disc and may also disconnect $\C$ into many components.
Disconnecting $\C$ is one of the main tools in our construction. It is done by the robot on purpose to allow it to find its way back to the marker and so that it is possible to recover the connectivity of $\C$ on the way back.
Such a walk back to the marker is the second action performed in the main loop and described as procedure \ref{alg:return-to-marker}.
The third action is done once the robot is back at the marker, and it  is given as procedure \ref{alg:extend-pseudo-square}.
This procedure extends the   rough disc, ensuring the property that there are no full cells at  a prescribed bounded distance from the   rough disc, except the marker.
While restoring this property, the robot may again disconnect some components but all of them are at a bounded distance from the   rough disc and thus the robot will be able to find them easily.
Additionally, it may happen that the cell brought to the   rough disc was the last cell of $\C$.
In such a case, as the last part of procedure \ref{alg:extend-pseudo-square}, the robot places the marker near another component close to the   rough disc, if one exists.
This will be the new free component $\C$ in which the robot will find the next brick in the next iteration of the main loop.
If no such $\C$ exists, then the robot adds the brick from the marker to the rough disc, thus completing the construction of the nest.

Many of the difficulties described above come from our desire to keep the complexity of the algorithm optimal, i.e., $O(sz)$. If complexity were not an issue, the following much simpler algorithm would be sufficient. The robot first builds a horizontal line at the level of a South-most cell of the initial field, by gradually squeezing down the field, keeping it connected at all times. Then it transforms the line into a nest. The squeezing down can be performed by iteratively repeating the following steps. First, the robot makes sure that it is not on the lowest level (if the line is not yet constructed, this can be done by finding a full cell with a full South neighbor). Then the robot goes to some North-West extremity of the field. Then it picks the brick from this cell and drops it one level down, ensuring the connectivity. When the field is squeezed down to a line, the robot will recognize this and easily transform it into a nest. 

This idea potentially requires time $\Theta(z)$ to lower a brick by one level. Since there are $z$ bricks possibly on $\Theta(s)$ levels, the entire time would be $\Theta(sz^2)$ in the worst case,  which is suboptimal. Thus we proceed with the detailed description of the optimal algorithm  \ref{alg:nest} whose high-level idea was described before.

\subsection{Moving bricks out of the way}

One of the challenges in constructing the   rough disc is to have enough room so that, while expanding, it does not merge with the remainder of the field.
This is one of the goals of procedure \ref{alg:sweeping}. Its 
high-level idea is the following.
It ensures an invariant that has to hold whenever the agent goes to retrieve the next brick in order to extend the   rough disc.
This invariant is that there are no full cells, other than the cell $\marker$ that is the marker, within a given bounded distance from the current   rough disc $\D$.
This procedure is called when the robot is at the marker,  in two situations.
The first one is right before the main loop: in this case the marker, the   rough disc (of size one) and its corresponding neighborhood occupy a constant number of cells and hence in this case the robot is able to decide whether a given cell is in $\D\cup\{\marker\}$.
The second situation occurs after each extension of the   rough disc.
In this case, the size of $\D$ may be unbounded but a walk around its border (which can be done with stop, using the marker) allows to determine if there is a full cell $c$ within a bounded distance from $\D$, that does not belong to the   rough disc itself.
Whenever such a full cell $c$ is found, the robot picks the brick from $c$ and searches for an empty cell at distance at least $\distPC$ from the   rough disc.
This searching walk is done in such a way as to ensure the return to the   rough disc. At the end of the procedure, the robot places the marker next to some free component.   
Below is the pseudo-code of  procedure \ref{alg:sweeping}. In this pseudo-code, we use the notion of the robot going \emph{in direction away} from the   rough disc $\D$.
This is the direction which strictly increases the distance between the robot and the   rough disc. In the case where there are two such directions, we use the priority North, East, South, West.

\begin{figure}[htb!]
\begin{center}
\begin{minipage}{1\linewidth}
\begin{algorithm}[H]
\SetAlgoRefName{Sweep}
\SetAlgorithmName{Procedure}{}{}
	\caption{sweeping away bricks that are close to the   rough disc}
	\label{alg:sweeping}
	$\marker:=$ the marker
	
	go to the closest cell of the   rough disc $\D$ \label{ln:sweeping:go-to-square}
	
	perform a full counterclockwise traversal of the border of $\D$, executing the following actions after each step: \label{ln:sweeping:counterclockwise-traversal}
	
	\Indp
	   \For{each full cell $c\notin \D\cup\{\marker\}$ at distance at most $\distPC$ from the robot \label{ln:sweeping:for}}
	   {
	       $c':=$ the cell currently occupied by the robot \label{ln:sweeping:c-prime-def}
	   
	       go to $c$ and pick the brick \label{ln:sweeping:pick-from-c}

	       move in direction away from $\D$ and stop at the first empty cell at distance at least $\distPC$ from $\D$ \label{ln:sweeping:go-to-drop}
	       
	       drop the brick and return to $c'$ \label{ln:sweeping:drop-and-return}
	   }
	
	\Indm
	\If{there exists a free component $\C$ \label{ln:sweeping:if-component}}
	{
	    pick the brick from marker and place it at distance $\distPM$ from the   rough disc and at distance $\distMC$ from $\C$, creating a new marker
	    \label{ln:reposition-marker}
	    
	}

	    go to the marker \label{ln:sweeping:last-trip-to-marker}

\end{algorithm}
\end{minipage}
\end{center}
\end{figure}

\subsection{Finding the next brick}

The high-level idea of procedure \ref{alg:find-next-brick} is the following.
In may happen that the cells that belong to a free component $\C$ and are close to the marker cannot be rearranged in such a way that the robot be able to obtain a brick that it can then use to extend the   rough disc and at the same time keep the connectivity of $\C$ and ensure that $\C$ remains close to the marker.
Thus, the robot has to retrieve the needed brick by following a potentially long walk; we will call it a \emph{search walk} and formally define it in Section~\ref{sec:search-walks}.
The search walk needs to be carefully chosen to ensure that it ends at a location where it is possible to find the desired brick and so that the robot be able to return to the marker.
Moreover, this walk has to be sufficiently  short to guarantee the time $O(sz)$ of the algorithm, i.e., the length of each walk must be $O(s)$.
We ensure the latter as follows: each search walk $\W$ leads alternately in two non-opposite directions, e.g., North and West.
The return is guaranteed by repeatedly performing an action called \emph{switch} while walking along $\W$, which we formally describe in Section~\ref{sec:ensuring-return}.
Intuitively speaking, the switch eliminates the cells adjacent to $\W$ at which the robot may incorrectly turn on its way back to the marker.
The switch potentially disconnects the component but the robot is able to recover the connectivity while backtracking along $\W$.
Finally, in Section~\ref{sec:obtaining-brick}, we describe how the desired brick can be found at the end of $\W$.

\subsubsection{Search walks} \label{sec:search-walks}

Suppose that the robot is currently at a full cell and there is a full cell in front of the robot.
A \emph{left-free} (respectively \emph{right-free}) segment consists of all full cells that will be visited by the robot that moves without changing its direction until one of the following conditions holds:
\begin{enumerate}[label={\normalfont(S\arabic*)},align=left,leftmargin=*]
\item\label{it:c:turn-or-leaf}  the robot arrives at a full cell that has an empty cell in front of it and has an empty cell to the left (respectively right) of it; such a segment is called \emph{terminal},
\item\label{it:c:corner} the robot arrives at the first special cell such that the cell to the left (respectively right) of the robot is full.
\end{enumerate}
Whenever the orientation is not important or clear from the context we will refer to a left-free or right-free segment by saying \emph{segment}.
Note that not every special cell terminates the above sequence of moves.

We now define a \emph{search walk} $\W$ of the robot in an arbitrary component $\C$ (cf. the example in Fig. \ref{fig:search-walk}).
A search walk depends on the initial position of the robot in  $\C$ and on its orientation.
We make two assumptions in the definition: the robot is initially located at a full cell of $\C$ and, if $|\C|>1$, then there is a full cell in front of the robot.
The search walk $\W$ is a concatenation of segments.
The first segment is both left-free and right-free.
If the cell $c$ at the end of this segment is a leaf, then the construction of $\W$ is completed.
Otherwise, note that there is a full cell to the left of the robot located at $c$ or to the right of it.
In the former case, the search walk is called \emph{left-oriented} and in the latter it is \emph{right-oriented}.
Intuitively, a left-oriented search walk prescribes going straight until it is possible to go left, then going straight until it is possible to go right, and so on, alternating directions, until a stop condition is satisfied. A similar intuition concerns right-oriented search walks. 

More formally, if $|\C|>1$, then the search walk $\W$ consists of a single cell. Otherwise, in a right-oriented (respectively left-oriented) search walk, the segments are sequentially added to $\W$, cyclically alternating the following construction steps.
\begin{enumerate}[label={\normalfont(W\arabic*)},align=left,leftmargin=*]
\item\label{it:segment1} The robot traverses a right-free (respectively left-free) segment, adding it to $\W$.
This segment becomes the last segment in $\W$ if it is a terminal.
If the segment is not the last one, then the robot turns right (respectively left).
\item\label{it:segment2} The robot traverses a left-free (respectively right-free) segment, adding it to $\W$.
This segment becomes the last segment in $\W$ if it is a terminal.
If the segment is not the last one, then the robot turns left (respectively right).
\end{enumerate}

\begin{figure}[htb!]
\begin{center}
\includegraphics[scale=0.8]{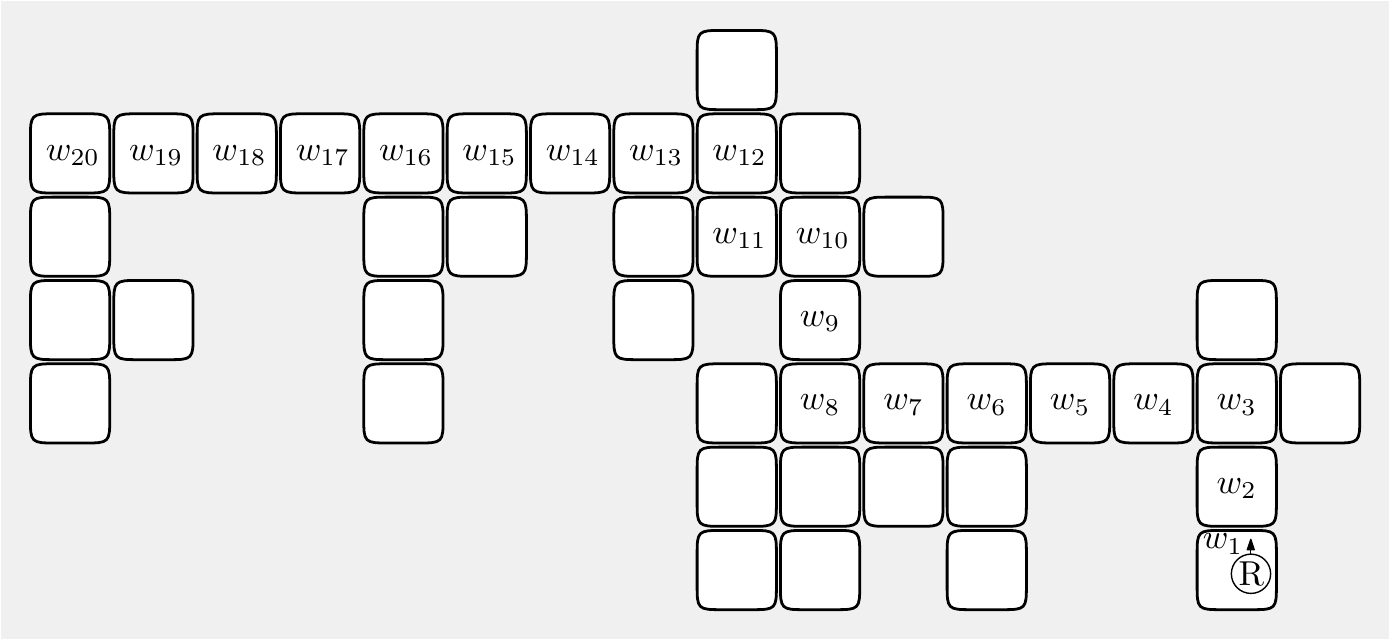}
\caption{An example of a search walk $\W=(w_1,\ldots,w_{20})$ that is constructed by the robot initially located at $w_1$ and facing North. This search walk is left-oriented, and has three left-free segments $S_1=(w_1,w_2,w_3)$, $S_3=(w_8,w_9,w_{10})$, $S_5=(w_{11},w_{12})$ and three right-free segments $S_2=(w_3,\ldots,w_8)$, $S_4=(w_{10},w_{11})$, $S_6=(w_{12},\ldots,w_{20})$.}
\label{fig:search-walk}
\end{center}
\end{figure}

\subsubsection{Ensuring the return from a search walk} \label{sec:ensuring-return}

We start with a high-level idea of the mechanism that will ensure the return from a search walk. 
Whenever the robot follows a search walk $\W$, it may a priori not be able to return to the origin of $\W$.
This is due to the fact that, e.g., if $\W$ is left-oriented, then any segment that is right-free may have an unbounded number of special cells such that each of them is adjacent to a cell that does not belong to $\W$.
Thus, the returning robot is not able to remember, using its bounded memory, which of such cells does not belong to the search walk and should be skipped.
To overcome this difficulty, the robot will make small changes in the field close to the search walk while traversing it for the first time.
These changes may disconnect the component in which the robot is walking, and this may result in creating many new components.
While doing so, we will ensure two properties.
First, thanks to the modifications in the field performed while traversing $\W$, the robot is able to return to the first cell of this search walk.
Second, while backtracking on $\W$, the robot is able to undo earlier changes and recover the connectivity of the component.

Each cell $c$ at which the robot stops to perform the above-mentioned modification will be called a \emph{break point} and is defined as follows.
First, we require that $c$ be an internal cell of a segment $S$, i.e., neither the first nor the last cell of $S$.
Second, if $S$ is left-free (respectively right-free), then when the robot traversing $S$ is at $c$, then there is a full cell $f$ to the right (respectively left) of it.
Clearly, the cell $e$ to the left (respectively right) of the robot is empty.
The following couple of actions performed by the robot located at such a cell $c$ are called a \emph{switch}: if the robot is not carrying a brick, then the robot brings the brick from $f$ and then places it at $e$, and if the robot is carrying a brick, then the robot places it at $e$ and then brings the brick from $f$.
Note that the switch may disconnect the component in which the robot is located thus creating two new components.
One new component is the one in which the robot is located and this is the component that contains the search walk.
The second new component is the one that contains the cell $f'$ adjacent to $f$ and at distance two from $c$, if $f'$ is full. If the cell $f'$ is full
and belongs to a separate component, then this second component containing $f'$ will be called a \emph{switch-component}.
Whenever the robot traversing a segment performs the switch at each internal special cell of the segment, we say that this is a \emph{switch-traversal}
(cf. the example in Fig. \ref{fig:switch-traversal}).
\begin{figure}[htb!]
\begin{center}
\includegraphics[scale=0.8]{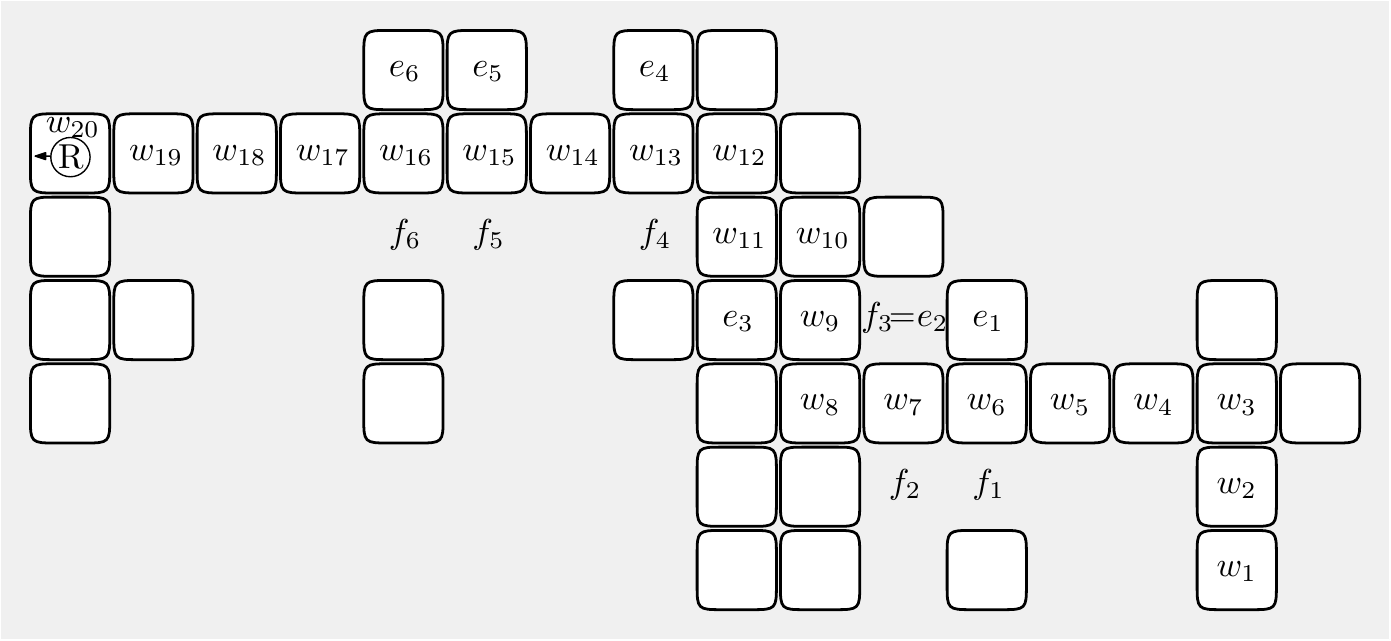}
\caption{The field from Figure~\ref{fig:search-walk} after switch-traversal of the search walk from Figure~\ref{fig:search-walk}. The cells $w_6,w_7,w_9,w_{13},w_{15}$ and $w_{16}$ are the break points at which the robot moves a brick from a cell $f_i$ to $e_i$ for $i\in\{1,\dots,6\}$. Note that a brick is moved from $f_2$ to $e_2$ while traversing the second segment and then the same brick is moved to $e_3$ while traversing the third segment.}
\label{fig:switch-traversal}
\end{center}
\end{figure}

\subsubsection{Obtaining a brick at the end of a search walk} \label{sec:obtaining-brick}

Informally, the purpose of traversing the entire search walk by a robot is to arrive at a location in the current component $\C$ of the robot, where
the robot can start a procedure aimed at obtaining a brick whose removal will not disconnect $\C$.
We will say that such a brick is \emph{free}.
In our algorithm, we check the condition \ref{it:c:turn-or-leaf} to learn if the last segment is terminal.
According to the condition, the terminal segment may end with a leaf, and in such a case the robot is at a cell with a free brick.
If the terminal segment does not end with a leaf, then there need not exist a free brick located in a close neighborhood of the robot.
However, we will prove that it is possible to perform a series of changes to the field that results in creating a configuration of bricks that does contain a free brick.

We now define the behavior of the robot that ended the switch-traversal of the last segment $S$ of a search walk $\W$ and arrived at a cell that is not a leaf.
The following series of moves is called \emph{shifting} (cf. Fig. \ref{fig:obtaining-brick}).
Suppose that the cell to the right (respectively left) of the robot is full. Note that this implies that $S$ is left-free (respectively right-free).
First the robot changes its direction so that a cell of $S$ is in front of it (i.e., the robot turns back).
The following three actions are performed until the stop condition specified in the third action occurs.
First, the robot picks the brick from the currently occupied cell.
Second, the robot moves one step forward --- thus backtracking along $S$.
Third, when the robot is at a special cell, then the shifting is completed, and otherwise the robot places the brick at the cell to the left (respectively right) of it.
Note that when the robot arrives at a special cell, it is carrying a brick and this is the desired free brick.
\begin{figure}[htb!]
\begin{center}
\includegraphics[scale=0.8]{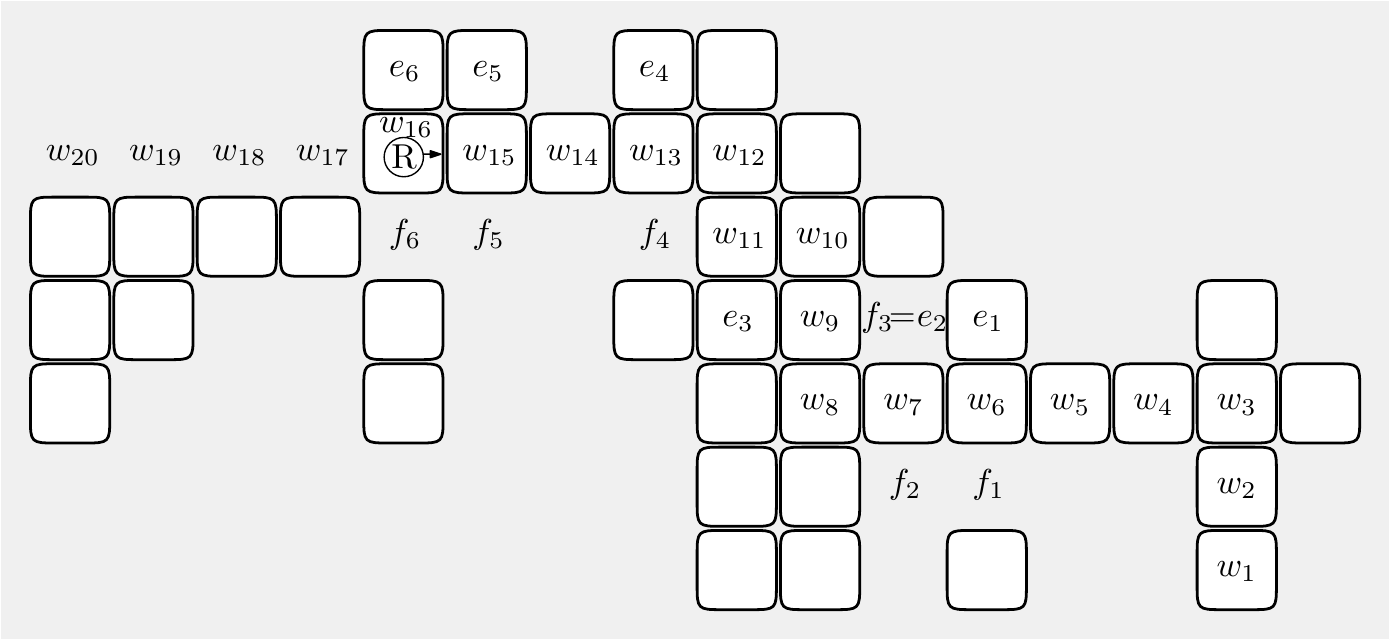}
\caption{The field from Fig.~\ref{fig:switch-traversal} at the end of shifting. The shifting ends with a right-free segment, at the cell $w_{16}$ because it has a full neighbor, the cell $e_6$. There is one fewer brick than in Fig. \ref{fig:switch-traversal} and this is the free brick obtained and carried by the robot.}
\label{fig:obtaining-brick}
\end{center}
\end{figure}

We now give a pseudo-code of procedure \ref{alg:find-next-brick} that obtains this brick.
\begin{figure}[htb!]
\begin{center}
\begin{minipage}{.8\linewidth}
\begin{algorithm}[H]
\SetAlgoRefName{FindNextBrick}
\SetAlgorithmName{Procedure}{}{}
	\caption{finding a free brick}
	\label{alg:find-next-brick}
		
	$\W:=$ the search walk that starts at the cell where the robot is located \label{ln:find-next-brick-W}
	
	go to the nearest cell belonging to a free component \label{ln:find-next-brick-go-to-marker}
	
	perform a switch-traversal of $\W$ \label{ln:find-next-brick-switch-traversal}
	
	\If{the robot is at a leaf \label{ln:find-next-brick-if-start}}
	{
	    pick the brick \label{ln:find-next-brick-if-brick}
	}
	\Else
	{
	    perform shifting \label{ln:find-next-brick-if-shifting}
	}
\end{algorithm}
\end{minipage}
\end{center}
\end{figure}

\subsection{Back to the marker}

Before presenting the high-level idea of  procedure \ref{alg:return-to-marker} that takes the robot carrying a brick back to the marker, we need the following definitions.
If $S$ is a segment, then the \emph{reversal} of $S$, denoted by $\reversal{S}$, is the segment composed of the same cells as $S$ but in the reversed order.
For a search walk $\W$ that is a concatenation of segments $S_1,\ldots,S_l$, define the \emph{reversal} of $\W$, denoted by $\reversal{\W}$, to be the walk that is the concatenation of segments $\reversal{S_l},\reversal{S_{l-1}},\ldots,\reversal{S_1}$, in this order.
Thus, following $\reversal{\W}$ means backtracking along $\W$, and in this section we give a procedure performing it, that reconnects the previous free component on the way.
We also define the orientation of $\reversal{\W}$ as follows.
If the last segment of $\W$ is left-free, then $\reversal{\W}$ is left-oriented, and otherwise $\reversal{\W}$ is right-oriented.
Thus, if $\W$ ended with a left-free (respectively right-free) segment, then $\reversal{\W}$ also starts with a left-free (respectively right-free) segment.

At a high level, the robot will perform a switch-traversal along $\reversal{\W}$, stopping at each break point to reconnect the corresponding switch-component with the component in which the robot is walking.
However, we need to ensure that, at the end, the robot stops at the right point, i.e., at the marker.
In order to ensure this, we define the following condition:
\begin{enumerate}[label={\normalfont(S\arabic*')},align=left,leftmargin=*]
\item\label{it:c:marker} the robot arrives at a cell at distance at most $\distMC$ from the marker.
\end{enumerate}
We define a \emph{return switch-traversal} of $\reversal{\W}$ to be a switch-traversal of $\reversal{\W}$ in which each verification of condition \ref{it:c:turn-or-leaf} is replaced by the verification of condition \ref{it:c:marker}.
Recall that the condition \ref{it:c:turn-or-leaf} is checked in the definition of a switch-traversal to determine the termination of a segment and consequently the termination of the entire search walk.
Intuitively, by replacing condition \ref{it:c:turn-or-leaf} with \ref{it:c:marker} we change the behavior of the robot so that it is looking for the marker while backtracking along $\W$, i.e., going along $\reversal{\W}$.

A high-level sketch of procedure \ref{alg:return-to-marker} is the following.
As indicated earlier, the robot essentially follows $\reversal{\W}$ and, as the return switch-traversal dictates, reconnects the switch components.
However, there is one special case in which the robot should not perform a switch while being at a cell $c'$ of $\reversal{\W}$, although $c'$ satisfies the definition of a break point.
This case occurs if the shifting moved all internal cells of the last segment of $\W$. In this case, 
it is enough for the robot to move to the next cell after $c'$ and start the return switch-traversal of $\reversal{\W}$ from there.
This is feasible because the cell $c$ and its neighbors are at a bounded distance from the robot when the shifting is completed.
Below is the pseudo-code of procedure  \ref{alg:return-to-marker}.

\begin{figure}[htb!]
\begin{center}
\begin{minipage}{1\linewidth}
\begin{algorithm}[H]
\SetAlgoRefName{ReturnToMarker}
\SetAlgorithmName{Procedure}{}{}
	\caption{going back to the marker}
	\label{alg:return-to-marker}
		
	$\W:=$ the search walk traversed in the last call to \ref{alg:find-next-brick} \label{ln:return:W}
	
	{let $S_1,\ldots,S_l$ be the segments in $\W$} \label{ln:return:S}

	\If{the robot is at the first cell of $S_l$ \label{ln:return:if-shifting-first}}
	{
	    turn towards the penultimate cell of $S_{l-1}$ \label{ln:return:turn}
	    
	    move $\min\{2,|S_{l-1}|-1\}$ cells forward \label{ln:return:omit-fake-switch}
	    
	    if $|S_{l-1}|=2$ and $l>2$, then make a turn towards the penultimate cell of $S_{l-2}$ \label{ln:return:second-turn}
	    
	}
	starting at the current location, perform a return switch-traversal of $\reversal{\W}$ \label{ln:return:switch-traversal}
	
	go to the marker \label{ln:return:go-to-marker}

\end{algorithm}
\end{minipage}
\end{center}
\end{figure}

\subsection{Extending the   rough disc}

The aim  of procedure \ref{alg:extend-pseudo-square} is triple: it adds a new brick to the current   rough disc in a specific place, and it calls procedure \ref{alg:sweeping} to extend, if necessary, the empty space around the   rough disc and to ensure that the marker is close to some free component.
Whenever procedure \ref{alg:extend-pseudo-square} is called, the following conditions will be satisfied: the robot is at the marker and it is carrying a brick.
Below is the pseudo-code of procedure \ref{alg:extend-pseudo-square}.

\begin{figure}[htb!]
\begin{center}
\begin{minipage}{.9\linewidth}
\begin{algorithm}[H]
\SetAlgoRefName{ExtendRoughDisc}
\SetAlgorithmName{Procedure}{}{}
	\caption{adding one brick to the rough disc $\D$}
	\label{alg:extend-pseudo-square}

	place the brick at the unique cell $e$ such that $\D \cup \{e\}$ is a rough disc  \label{ln:pseudo-square:extend}
	    
	call procedure \ref{alg:sweeping} \label{ln:pseudo-square:sweeping}

\end{algorithm}
\end{minipage}
\end{center}
\end{figure}

Now the pseudo-code of the entire algorithm can be succinctly formulated as follows.

\begin{figure}[htb!]
\begin{center}
\begin{minipage}{.9\linewidth}
\begin{algorithm}[H]
\SetAlgoRefName{Nest}
	\caption{building a nest from any connected field}
	\label{alg:nest}

	\If{the span of the field is at most $2$}
	{
	    exit \Comment{the field is already a nest}
	}
	
	the cell occupied by the robot becomes the marker

	a full cell at distance 2 from the marker becomes the initial rough disc
	
	call procedure \ref{alg:sweeping} \label{ln:nest:init-sweeping}
	
	\While{there exists a free component $\C$}
	{	    
	    call procedure \ref{alg:find-next-brick} \label{ln:nest:find-next-brick}
	    
	    call procedure \ref{alg:return-to-marker} \label{ln:nest:return-to-marker}
	    
	    call procedure \ref{alg:extend-pseudo-square} \label{ln:nest:extend-pseudo-square} \Comment{moves the marker if necessary}
	}

	pick the brick and place it at the unique cell $e$ of the rough disc $\D$ such that $\D \cup \{e\}$ is a rough disc \label{ln:nest:final-brick}

\end{algorithm}
\end{minipage}
\end{center}
\end{figure}

The following is the main result of this paper.

\begin{theorem} \label{thm:main}
Algorithm \ref{alg:nest} builds a nest starting from any connected field of size $z$ and span $s$ in time $O(sz)$.
This time is worst-case optimal.
\end{theorem}

\section{Analysis of the algorithm}

This section is devoted to the proof of Theorem \ref{thm:main}.

\subsection{Search walks} \label{sec:analysis:search-walk}

The first three observations provide simple properties of search walks that will be used in the sequel.

\begin{observation} \label{obs:after-finding-brick}
Suppose that $\W$ is the search walk from procedure \ref{alg:find-next-brick}.
After performing the switch-traversal in line~\ref{ln:find-next-brick-switch-traversal} of the procedure, the robot is at the end of $\W$.
\qed
\end{observation}

\begin{observation} \label{obs:side-reversal}
A left-free (respectively right-free) segment becomes right-free (respectively left-free) as a result of a switch-traversal.
\qed
\end{observation}

\begin{observation} \label{obs:end-of-W}
Let $\W$ be a search walk whose last segment $S$ is left-free (respectively right-free).
Upon completion of a switch-traversal of $\W$, if the robot is not at a leaf, then the cell to the right (respectively left) of the robot is full and the cells in front of it and to the left (respectively right) of it are empty.
\qed
\end{observation}

We will need some additional notation.
Let $\W$ be a search walk in a component $\C$ and $b$ be the last break point of $\W$.
All cells of $\W$ starting at $b$ are called the \emph{prefix of} $\reversal{\W}$.
Note that, at each point of the execution of procedure \ref{alg:find-next-brick}, some parts of $\C$ become the switch-components.
The other component, namely the one that overlaps with $\W$ will be called the \emph{$\W$-component}.
(Note that we specify that it overlaps with $\W$ and not contains $\W$ since, after the shifting, some cells of $\W$ that are in the prefix of $\reversal{\W}$ are empty and thus are not a part of the $\W$-component.)
Upon completion of shifting, each component that contains a brick placed during shifting will be called a \emph{shift-component}.

\begin{lemma} \label{lem:shifting}
Upon completion of shifting performed in procedure \ref{alg:find-next-brick}, we have the following properties.
\begin{enumerate}[label={\normalfont(\roman*)},align=left,leftmargin=*]
\item\label{it:shifting:unique}
The shift-component is unique.
\item\label{it:shifting:shift-component}
Let $c'$ be the cell adjacent to the cell occupied by the robot and also adjacent to the cell on which the last brick has been placed during shifting.
If $c'$ is full, then it belongs to the shift-component and it belongs to $\W$.
If $c'$ is empty, then it is adjacent to a cell of the shift-component.
\item\label{it:shifting:penultimate}
The robot is located at the penultimate special cell of $\W$ and is carrying a brick.
\end{enumerate}
\end{lemma}
\begin{proof}
Consider a shifting that occurs after a switch-traversal of $\W$.
Let $S$ be the last segment of $\W$.
The segment $S$ is either left-free or right-free, by the definition of a search walk.
We consider the former case (cf. the example in Fig. \ref{fig:shifting}); the latter one is analogous and will be skipped.
\begin{figure}[htb!]
\begin{center}
\includegraphics[scale=0.7]{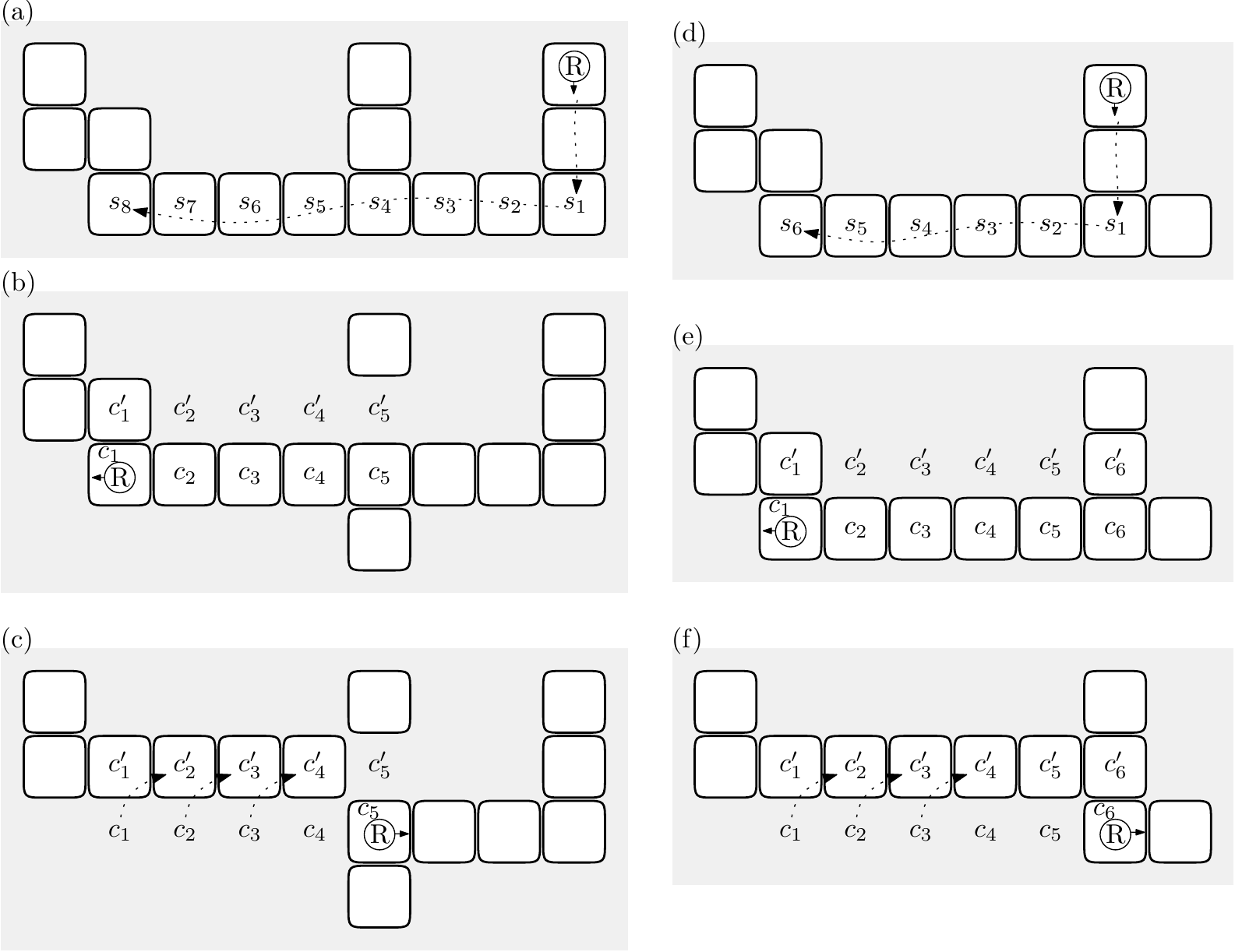}
\caption{(a)-(c) illustrates the case when the cell $c'$ in Lemma~\ref{lem:shifting}\ref{it:shifting:shift-component} is empty: (a) shows the robot traversing the penultimate segment after which, in (b), a switch-traversal of the last segment $(s_1,\ldots,s_8)$ is completed, and (c) depicts the field at the end of shifting. Note that in this case $c'=c_l=c_5'$.
(d)-(f) illustrates the case when $c'$ is full: (d) and (e) again depict the switch-traversal of the last two segments, and (f) shows the field at the end of shifting. In this case $c'=c_l=c_6'$ and this cell belongs to the penultimate segment of the search walk. During the traversal of the last segment no switch occurred.
As we will prove, the two latter properties are always satisfied together.
}
\label{fig:shifting}
\end{center}
\end{figure}

By Observation~\ref{obs:after-finding-brick} and the formulation of procedure \ref{alg:find-next-brick}, when the shifting starts the robot is at the last cell of $S$.
Denote this cell by $c_1$.
Let $c_2,\ldots,c_l$ be the cells of $S$ such that $c_i$ is adjacent to $c_{i-1}$ for each $i\in\{2,\ldots,l\}$, $c_l$ is a special cell in $\W$ and $c_2,\ldots,c_{l-1}$ are not special cells in $\W$.
Note that the latter sequence may be empty if two last cells of $S$ are special cells.
Denote by $c_{i}'$ the cell that is to the left of the robot when it is at $c_i$ during shifting.

According to the definition of shifting, the following occurs: for each $i\in\{1,\ldots,l-1\}$, the brick from $c_i$ is moved to the cell $c_{i+1}'$.
We argue that shifting is feasible, i.e., that the cell $c_i$, $i\in\{1,\ldots,l-1\}$, is full when the robot picks the brick from it and that the cell $c_{i+1}'$, $i\in\{1,\ldots,l-2\}$, is empty when the robot then drops the brick at it, for $i\in\{1,\ldots,l-1\}$.
The former is due to the fact that $c_i$ belongs to a segment which by definition is a sequence of full cells.
The latter holds because $c_2,\ldots,c_{l-1}$ are not special cells, i.e., their neighbors outside of $S$, in particular the cells $c_2',\ldots,c_{l-1}'$, are empty.

We now analyze what types of components are created during shifting from the component $\C$ in which it started.
Consider the components $\C_1,\ldots,\C_t$ obtained from $\C$ by making the cells $c_1,\ldots,c_{l-1}$ empty.
These are all cells from which the robot picked the bricks during shifting.
First observe that there are exactly two components as above, i.e., $t=2$.
The first one, say $\C_1$, contains the cell $c_1'$ because, according to the `if' statement in line~\ref{ln:find-next-brick-if-start} of procedure \ref{alg:find-next-brick}, $c_1$ is not a leaf and hence, due to Observation~\ref{obs:end-of-W}, the cell $c_1'$ is full.
The second component, let it be $\C_2$, contains $c_l$, which is full because it belongs to $\W$ and has not been emptied during shifting.
To finish the argument that $t=2$, note that by definition of shifting, the cells $c_2,\ldots,c_{l-1}$ are not special cells and thus have no full neighbors outside of $S$, and by Observation~\ref{obs:end-of-W}, the cell $c_1$ has no other full neighbors in $\C$ except $c_2$ and $c_1'$.
Observe that making the cells $c_2',\ldots,c_{l-1}'$ full during shifting glues these cells to $\C_1$, resulting in a single component.
This is the unique shift-component. Thus property \ref{it:shifting:unique} is proved.

The cell that is to the right of the robot upon completion of shifting is $c_l'$.
Note that $c'=c_l'$ and this cell is adjacent to $c_{l-1}'$.
Hence, if $c_l'$ is empty, then \ref{it:shifting:shift-component} holds.
Suppose that $c_l'$ is full.
The shift-component $\C_1$ and the $\W$-component $\C_2$ are the same: $c_l'$ is adjacent to the full cell $c_{l-1}'$ of $\C_1$ and is also adjacent to the full cell $c_l$ that belongs to $\C_2$.
Since $S$ is left-free, by the definition of switch-traversal we obtain that the brick located at $c_l'$ has not been dropped at $c_l'$ as a result of a switch.
Moreover, it cannot be a neighbor of an internal cell of $S$.
Since this is not the last cell of $S$ because $l>1$, we have that $c_l$ is the first cell of $S$.
Hence $c_l'$ is the penultimate segment of $\W$.
This proves property \ref{it:shifting:shift-component}.

Finally, the brick picked from $c_{l-1}$ is carried by the robot at the end of shifting.
At the end of shifting, after picking the brick, the robot moves to $c_l$ which is the penultimate special cell of $\W$ because the cells $c_2,\ldots,c_{l-1}$ are not special cells by definition.
This proves property \ref{it:shifting:penultimate} and completes the proof of the lemma.
\end{proof}

We say that the   rough disc has a \emph{gap of width $k$} if each free component is at distance $k+1$ from the   rough disc.
We say that a field is \emph{structured} if it satisfies the following conditions:
\begin{enumerate}[label={\normalfont(F\arabic*)},align=left,leftmargin=*]
 \item\label{it:structured-MPC} the marker is at distance $\distPM$ from the   rough disc and at distance $\distMC$ from some free component,
 \item\label{it:structured-gap} the   rough disc has a gap of width $\distPC$.
\end{enumerate}

The next lemma characterizes the situation upon completion of procedure  \ref{alg:find-next-brick}.

\begin{lemma} \label{lem:find-brick-correct}
Suppose that the field is structured at the beginning of an execution of procedure \ref{alg:find-next-brick}.
Let $\W$ be the search walk traversed by the robot in a free component $\C$ during this execution.
At the end of this execution we have the following properties:
\begin{enumerate}[label={\normalfont(\roman*)},align=left,leftmargin=*]
\item\label{it:find-robot} the robot is located in the prefix $\reversal{\W}$ and carries a brick.
\item\label{it:find-components} each brick that belongs to $\C$ at the beginning of the execution is either in a switch-component, in the shift-component or in the $\W$-component,
\item\label{it:find-gap} if $\W$ starts with a cell at distance $\distPC$ from the   rough disc, then the   rough disc has a gap of width at most $\distPC$ and at least $\distPCmin$.
\end{enumerate}
\end{lemma}
\begin{proof}
Consider the robot at the end of the execution of procedure \ref{alg:find-next-brick}.
We consider two cases depending on the execution of the `if' instruction in line \ref{ln:find-next-brick-if-start} of procedure \ref{alg:find-next-brick}.
In the first case, suppose that the instruction in line~\ref{ln:find-next-brick-if-brick} is executed.
The robot is at the last cell $c$ of $\W$ and carries a brick that it picked from $c$.
This proves~\ref{it:find-robot} in this case.
Since $c$ was a leaf upon arrival of the robot at $c$, picking $c$ did not disconnect the $\W$-component, which implies that the shift-component equals the $\W$-component and consequently proves \ref{it:find-components} of the lemma.

In the second case, suppose that instruction in line~\ref{ln:find-next-brick-if-shifting} of procedure \ref{alg:find-next-brick} is executed, i.e., the robot performed shifting.
Then, statement \ref{it:find-robot} follows from Lemma~\ref{lem:shifting}\ref{it:shifting:penultimate}, since, by definition, the prefix $\reversal{\W}$  contains  all cells between the last special cell and the penultimate special cell of $\W$, inclusively.
Again by Lemma~\ref{lem:shifting}\ref{it:shifting:unique}, the shift-component is unique which implies statement \ref{it:find-components} in this case.

We now prove \ref{it:find-gap}.
Since the field is structured at the beginning of the execution of procedure \ref{alg:find-next-brick}, it has a gap of width $\distPC$.
The first cell of $\W$ does not change its position and since its distance from the   rough disc is $\distPC$, we have that the gap is not larger than $\distPC$ at the end of the execution.
Now we analyze how much the gap can decrease.
During the execution of procedure \ref{alg:find-next-brick}, the robot performs two types of changes in the field.
The first one is a switch, which places each brick next to a full cell of the component.
Thus, this can decrease the gap by at most one.
The second one is the shifting, which moves each brick to a cell at distance $2$ from its origin, which may decrease the gap by at most $2$.
The bricks moved during a switch do not belong to $\W$ and the bricks moved during shifting belong to it. This proves \ref{it:find-gap} and concludes the proof of the lemma.
\end{proof}

\subsection{Reversals of search walks} \label{sec:analysis:reversal}

If a component is at distance larger than $\distPC$ from the   rough disc, then it is called a \emph{lost component}.
Recall the definition of a switch, in which the robot located at a cell $c$ brings a brick from a full cell $f$ and places it at an empty cell $e$.
The cell $e$ will be called a \emph{switch-cell}.
If, during a return switch-traversal, the robot is backtracking on a segment $S$, which was added to the walk using condition \ref{it:segment1} or \ref{it:segment2}, and it checks, upon arrival at a special cell $c$ of $S$, whether the cell to the left (respectively right) of $c$ is full, then we say that the robot is \emph{looking for left} (respectively \emph{right}) \emph{turn} during this traversal of $S$.
Suppose that the robot is located at a cell $c$ of a segment $S$ during backtracking.
We say that the robot is \emph{oriented} on $\reversal{S}$ if it is directed towards the beginning of $S$.
When the robot backtracks on $S$, then all cells that belong to $S$ and are located between $c$ (inclusively) and the beginning of $S$ (exclusively) are called the \emph{leftover of} $S$.
In other words, these cells are the ones to be yet visited by the robot backtracking on $S$.

In the three following lemmas we prove that procedure \ref{alg:return-to-marker} can be correctly executed by the robot.

\begin{lemma} \label{lem:skipping-fake-switch}
The robot can correctly execute lines \ref{ln:return:W}-\ref{ln:return:second-turn} of procedure \ref{alg:return-to-marker}.
\end{lemma}
\begin{proof}
Note that lines \ref{ln:return:W} and \ref{ln:return:S} result in no actions of the robot and only introduce notation for the description of the remaining steps of the procedure.
The condition in the `if' statement in line~\ref{ln:return:if-shifting-first} can be checked using Lemma~\ref{lem:shifting}\ref{it:shifting:shift-component}: the robot tests whether the cell $c'$ in the lemma is full or not and it is full if and only if the robot is at the first cell of $S_l$.
Now we analyze the instructions within the `if' statement.
In line~\ref{ln:return:turn}, the penultimate cell of $S_{l-1}$ is the above cell $c'$ and hence the robot can perform the correct turn.
The value of $|S_{l-1}|$, may be large and thus unknown to the robot.
However, for calculating $\min\{2,|S_{l-1}|-1\}$ it is enough to know whether $|S_{l-1}|=2$ or $|S_{l-1}|>2$.
This can be remembered by the robot while performing the traversal of a search walk in procedure \ref{alg:find-next-brick} (although the robot does not know whether the current segment is penultimate, it may always remember the information about the two recently traversed segments).
To execute the instruction in line~\ref{ln:return:second-turn}, the robot can remember whether $l>2$, and it turns in the opposite direction to the turn just made in line~\ref{ln:return:turn}, which follows from the definition of a search walk.
\end{proof}

\begin{lemma} \label{lem:return-segment}
Suppose that the robot is oriented on $\reversal{S}$, where $S$ is a left-free (respectively right-free) segment of some search walk.
If:
\begin{enumerate}[label={\normalfont(\roman*)},align=left,leftmargin=*]
\item\label{it:return-segment-a1} the robot previously performed a switch-traversal along $S$,
\item\label{it:return-segment-a2} the robot is looking for a left (respectively right) turn while traversing $\reversal{S}$,
\end{enumerate}
then, during the return switch-traversal, the robot arrives at the end of $\reversal{S}$ and turns left (respectively right).
\end{lemma}
\begin{proof}
Assume that the robot is looking for a left turn.
The other case is analogous and we omit it.
Thus, $S$ is left-free prior to the switch-traversal performed in \ref{it:return-segment-a1}.
By Observation~\ref{obs:side-reversal}, $\reversal{S}$ is also left-free after the switch-traversal.
Let $S'$ be the maximal sequence of cells that the robot traversed during the return switch-traversal while looking for a left turn, before it turned left.
Since $\reversal{S}$ is left-free, this implies that the robot reaches the end of $\reversal{S}$, i.e., $S'$ contains the endpoint of $\reversal{S}$.
Let $c'$ be the cell at which the robot finishes the traversal of $S'$.
Let $c$ be the first cell of $S$.
We need to show that $c=c'$, i.e., that the robot does make a turn when located at $c$.
By the definition of a search walk, the first cell of $S$ is a special cell.
By the definition of a switch-traversal, the robot does not perform a switch at the end nor at the beginning of a segment.
Thus, in particular, no switch occurred when the robot was located at $c$ during the traversal of $S$ while performing the search walk $\W$.
This means that a neighbor of $c$ that is not in $S$ is full during the traversal of $S$ in $\W$ if and only if this neighbor is full during the traversal of $S'$.
Thus, $c$ is a special cell when the robot arrives at $c$ while traversing $S'$.
Moreover, this special cell ends the segment $p(S)$ that precedes $S$ in $\W$.
Since $p(S)$ is not a terminal segment, we have that it ended by satisfying condition \ref{it:c:corner}.
Since $S$ is left-free, the segment $p(S)$ is right-free by the definition of a search walk.
This implies that when the robot arrived at $c$ during the search walk $\W$, it had a full cell to its right.
This full cell belongs to $S$ and therefore, when the robot arrives at $c$ while traversing $S'$, it has a full cell to its left.
Since the robot is looking for a left turn, $c$ becomes the last cell of $S'$ proving that $c=c'$, and moreover, the robot turns left at $c$, as required by the lemma.
\end{proof}

The following lemma shows that, while executing procedure \ref{alg:return-to-marker}, the robot does not loose its way during the backtrack on a previously traversed search walk.

\begin{lemma} \label{lem:return-backtracking}
Let $\W$ be a search walk.
Suppose that the robot traversed $\W$ by executing procedure \ref{alg:find-next-brick} and then executed procedure \ref{alg:return-to-marker}.
These two executions ensure that the robot traverses $\reversal{\W}$ after the traversal of $\W$.
\end{lemma}
\begin{proof}
Denote by $S_1,\ldots,S_l$ the segments whose concatenation, in this order, gives the search walk $\W$.
By definition, a return switch-traversal is a walk that traverses a sequence of segments denoted by $S_1',\ldots,S_{l'}'$.
We argue by induction on $i\in\{1,\ldots,l'\}$ that:
\begin{enumerate}[label={\normalfont(\alph*)},align=left,leftmargin=*]
 \item\label{induction:reversal} $S_i'=\reversal{S_{l-i+1}}$, and
 \item\label{induction:turn} upon arrival at the end of $S_i'$, the robot becomes oriented on $S_{l-i}$.
\end{enumerate}
Before giving the inductive proof we note that property \ref{induction:reversal} implies that the return switch traversal is a concatenation of segments $\reversal{S_l},\ldots,\reversal{S_1}$ and thus it is the walk $\reversal{\W}$ as required by the lemma.
We note that property \ref{induction:turn} is used for technical reasons to conduct the inductive argument.

For the base case of $i=1$ we need to argue that the robot traverses the segment $\reversal{S_l}$ and makes the turn after which the traversal of $\reversal{S_{l-1}}$ will start.
By Lemma~\ref{lem:find-brick-correct}\ref{it:find-robot}, the robot is at a cell $c'$ in the segment $\reversal{S_l}$ at the end of execution of procedure \ref{alg:find-next-brick}.
Thus, the robot is in $\reversal{S_l}$ when procedure \ref{alg:return-to-marker} starts.
The segment $S_l$ has been appended to $\W$ in construction step \ref{it:segment1} or \ref{it:segment2} as a left-free or a right-free segment.
We assume that $S_l$ is left-free. The proof for a right-free segment $S_l$ is analogous and will be omitted.
If shifting did not occur during the execution of procedure \ref{alg:find-next-brick} then, according to line~\ref{ln:find-next-brick-if-brick}, the robot is at the end of the search walk $\W$, i.e., at the end of $S_l$.
If shifting did occur in line~\ref{ln:find-next-brick-if-shifting} of procedure \ref{alg:find-next-brick}, then $c'$ is not the last cell of $S_l$.

If $c'$ is the first cell of $S_l$ (which is the last cell of $\reversal{S_l}$), then by Lemma~\ref{lem:skipping-fake-switch}, the robot can correctly verify this in line~\ref{ln:return:if-shifting-first} of procedure \ref{alg:return-to-marker}.
Moreover, this implies that \ref{induction:reversal} holds for $i=1$.
Then, the robot performs the turn in line~\ref{ln:return:turn} which proves \ref{induction:turn}.

Now suppose that $c'$ is not the first cell of $S_l$.
The assumption \ref{it:return-segment-a1} of Lemma~\ref{lem:return-segment} is satisfied, where we take $S=S_l$ in the lemma.
By the definition of return switch-traversal, the fact that $\W$ ended with a left-free segment (recall that $S_l$ is left-free) implies that, when the robot starts its walk in line~\ref{ln:return:switch-traversal} of procedure \ref{alg:return-to-marker}, it follows a left-oriented walk.
The traversal of $S_1'$ is thus determined by the first construction step in \ref{it:segment1} for this left-oriented walk.
Thus, the robot traverses a left-free segment as the first segment $S_1'$, i.e., it is looking for a left turn.
This implies that the assumption \ref{it:return-segment-a2} of Lemma~\ref{lem:return-segment} is satisfied.
Then, \ref{induction:reversal} and \ref{induction:turn} follow from Lemma~\ref{lem:return-segment}.
This completes the proof of the base case.

For the proof of the inductive step suppose that assertions \ref{induction:reversal} and \ref{induction:turn} hold for some $i\in\{1,\ldots,l-1\}$ and we prove them for $i+1$. By the inductive assumption, the robot is at the end of $\reversal{S_{l-i+1}}$  when procedure \ref{alg:return-to-marker} starts.
We consider three cases.

Case 1. $i+1=2$, $S_{l-1}$ is of length $2$, and the robot is at the beginning of $S_l$ upon completion of shifting performed in line~\ref{ln:find-next-brick-if-shifting} of procedure \ref{alg:find-next-brick}.

By Lemma~\ref{lem:skipping-fake-switch}, the condition in the `if' statement in line~\ref{ln:return:if-shifting-first} of procedure \ref{alg:return-to-marker} is correctly verified by the robot, and by inductive assumption, it is oriented on $\reversal{S_{l-1}}$.
Since the length of $S_{l-1}$ is $2$, the robot makes one step forward (see line \ref{ln:return:omit-fake-switch}) and makes a turn (see line \ref{ln:return:second-turn}).
This ensures claims \ref{induction:reversal} and \ref{induction:turn} in this case.

Case 2.  $i+1=2$, $S_{l-1}$ is of length larger than $2$, and the robot is at the beginning of $S_l$ upon completion of shifting performed in line~\ref{ln:find-next-brick-if-shifting} of procedure \ref{alg:find-next-brick}.

By the inductive assumption, the robot is oriented on $S_{l-1}$.
By Lemma~\ref{lem:skipping-fake-switch}, the robot makes two steps along $\reversal{S_{l-1}}$ by executing instruction in line \ref{ln:return:omit-fake-switch} of procedure \ref{alg:return-to-marker} and thus it is at the third cell of $\reversal{S_{l-1}}$ when the return switch-traversal in line \ref{ln:return:switch-traversal} starts.
Thus, since the robot is at a cell of $S_{l-1}$,  condition \ref{it:return-segment-a1} of Lemma~\ref{lem:return-segment} is satisfied.
By arguments analogous to the ones in the base case of induction, condition \ref{it:return-segment-a2} holds as well.
Thus by Lemma~\ref{lem:return-segment}, claims \ref{induction:reversal} and \ref{induction:turn} hold for $i+1$.

Case 3. $i+1> 2$ or ($i+1=2$ and the robot is not at the beginning of $S_l$ upon completion of shifting performed in line~\ref{ln:find-next-brick-if-shifting} of procedure \ref{alg:find-next-brick}).

The robot is at the beginning of $S_{l-1}$ when the return switch-traversal starts.
As before,  conditions \ref{it:return-segment-a1} and \ref{it:return-segment-a2} of Lemma~\ref{lem:return-segment} are satisfied.
Thus by Lemma~\ref{lem:return-segment}, claims \ref{induction:reversal} and \ref{induction:turn} hold for $i+1$.
\end{proof}

We say that the field is \emph{strongly structured} if the following conditions are satisfied:
\begin{enumerate}[label={\normalfont(T\arabic*)},align=left,leftmargin=*]
 \item\label{it:strongly-structured:structured} the field is structured,
 \item\label{it:strongly-structured:no-lost} there are no lost components, and
 \item\label{it:strongly-structured:marker} the robot is at the marker.
\end{enumerate}
We will show that the field is strongly structured at the beginning of any iteration of the `while' loop in Algorithm \ref{alg:nest}.

The following lemma shows that while backtracking on a search walk, the robot reconnects all switch components that were disconnected during 
procedure \ref{alg:find-next-brick}.
Moreover, this backtracking stops when the robot is at the marker.

\begin{lemma} \label{lem:return-nothing-lost}
Suppose that at the beginning of an iteration of the ``while'' loop in algorithm \ref{alg:nest} the field is strongly structured.
Then there are no lost components in the field and the robot is at the marker after the execution of procedures \ref{alg:find-next-brick} and \ref{alg:return-to-marker}.
\end{lemma}
\begin{proof}
Since, by assumption, the robot is at the marker at the beginning of the iteration of the `while' loop, the robot is at the marker when procedure \ref{alg:find-next-brick} starts.
Thus, in line \ref{ln:find-next-brick-go-to-marker} of this procedure, the robot goes to a free component $\C$.
A free component exists due to the condition of the loop.
The fact that such a component $\C$ is at distance $\distMC$ from the marker follows from the assumption that the field is structured.
Let $\W$ be the search walk traversed in a component $\C$ during the execution of procedure \ref{alg:find-next-brick}.
By Observation~\ref{obs:after-finding-brick}, the robot arrives at the last cell of $\W$ as a result of the switch-traversal in line~\ref{ln:find-next-brick-switch-traversal} of procedure \ref{alg:find-next-brick}.
By Lemma~\ref{lem:find-brick-correct}\ref{it:find-components}, each component obtained from $\C$ is either a $\W$-component, a switch-component or a shift-component.
By Lemma~\ref{lem:shifting}\ref{it:shifting:unique}, the shift-component is unique.
The shift-component equals the $\W$-component if shifting did not occur during the execution of procedure \ref{alg:find-next-brick}, i.e., when line~\ref{ln:find-next-brick-if-brick} of this procedure has been executed.
We now consider the case when the shifting did occur.
We have two possibilities, depending on whether the cell $c'$ from Lemma~\ref{lem:shifting}\ref{it:shifting:shift-component} is full or empty upon completion of shifting, i.e., at the beginning of the execution of procedure \ref{alg:return-to-marker}.
By Lemma~\ref{lem:shifting}\ref{it:shifting:shift-component}, if $c'$ is full, then it belongs to the shift component and to $\W$.
Thus, the shift-component is identical to the $\W$-component.
By Lemma~\ref{lem:shifting}\ref{it:shifting:shift-component}, if $c'$ is empty, then $c'$ is adjacent to a cell $c''$ of the shift component.
Since $c$ is a special cell, there is a full cell $f$ to the left or to the right of the robot.
Since $c'$ is to the left or to the right of the robot and $c'$ is empty, we have that $f\neq c'$.
Moreover, by Lemma~\ref{lem:shifting}\ref{it:shifting:shift-component}, $c$ is not the first cell of the last segment in $\W$.
Thus, the condition in the `if' statement in line~\ref{ln:return:if-shifting-first} is false and hence the robot starts the return switch-traversal in line~\ref{ln:return:switch-traversal} in procedure \ref{alg:return-to-marker} at the cell $c$.
In this switch-traversal, the robot moves the brick from $f$ to $c'$ thus reconnecting the shift component (whose full cell $c''$ is adjacent to $c'$) with the $\W$-component.
At this point, the shift component and the $\W$-component are equal and therefore the initial field $\C$ is transformed into a field with the $\W$-component and possibly many switch-components.

By Lemma~\ref{lem:return-backtracking}, the robot backtracks on $\W$ by executing procedure \ref{alg:return-to-marker}.
In other words, the robot traverses $\reversal{\W}$.
Since the robot traverses the entire search walk $\reversal{\W}$, by the definition of return switch-traversal, each switch-component is reconnected with the $\W$-component.
It remains to argue that during the return switch-traversal no lost component is created.
Such a component could be created, if there existed a full cell $c'$ adjacent to a cell $c$ of $\reversal{\W}$ such that the robot performs a switch at $c$.
This switch would remove the brick from $c'$ thus creating a lost component. We now show that this is impossible.
The cell $c'$ has been made full during the switch-traversal in line~\ref{ln:find-next-brick-switch-traversal} of procedure \ref{alg:find-next-brick}.
Thus, this happened during shifting in line~\ref{ln:find-next-brick-if-shifting}.
Since during shifting, after moving each brick, the robot makes one step forward along $\reversal{\W}$, $c$ does not belong to the first segment of $\reversal{\W}$.
By the definition of a search walk, $c$ belongs to the second segment $S$ of $\reversal{\W}$, and $c$ is the second cell in $S$.
By Lemma~\ref{lem:skipping-fake-switch}, the robot correctly executed the `if' instruction in line~\ref{ln:return:if-shifting-first} of procedure \ref{alg:return-to-marker}.
If $S$ is of length $2$, then the robot arrives at $c$ in line~\ref{ln:return:omit-fake-switch} and makes a turn in line~\ref{ln:return:turn}, and consequently the switch does not occur at $c$.
If $S$ is of length larger than $2$, then the robot makes two moves forward along $S$ in line~\ref{ln:return:omit-fake-switch}: the first of these moves brings the robot to $c$ and thus again no switch occurs when the robot is at $c$.
Hence we have proved that no lost component is created while backtracking on $\W$.

We finally prove that the robot is at the marker at the end of procedure \ref{alg:return-to-marker}.
By Lemma~\ref{lem:return-backtracking}, the robot traverses $\reversal{W}$ in line~\ref{ln:return:switch-traversal} of this procedure.
Thus, it becomes located at its end, which is the beginning of $\W$, and hence it can correctly check that it is at distance $\distMC$ from the marker, by the definition of switch-traversal.
More precisely, this is guaranteed by condition~\ref{it:c:marker} that terminates $\reversal{\W}$. Then the robot gets to the marker executing line \ref{ln:return:go-to-marker} of procedure \ref{alg:return-to-marker}.
\end{proof}

\subsection{Analysis of the   rough disc construction} \label{sec:analysis:pseudo-square}

The following lemma shows that procedure \ref{alg:sweeping} does not create lost components and makes sufficient room around the   rough disc.
Moreover, the cost of each call of the procedure is $O(s)$.

\begin{lemma} \label{lem:sweeping}
Suppose that the robot is at the marker and that there are no lost components.
An execution of procedure \ref{alg:sweeping} takes time $O(s)$ and results in a strongly structured field.
\end{lemma}
\begin{proof}
We first prove that the field is structured after executing the procedure.
The robot can correctly go to the nearest cell of the   rough disc $\D$ in line~\ref{ln:sweeping:go-to-square}.
Indeed, in a call to procedure \ref{alg:sweeping} in line~\ref{ln:nest:init-sweeping} of algorithm \ref{alg:nest} this is because $|\D|=1$ and $\D$ is at distance $2$ from the robot.
In each call made in procedure \ref{alg:extend-pseudo-square} this is because the robot is at the marker and, by Lemma~\ref{lem:sweeping}\ref{it:find-gap}, the   rough disc and the marker form two separate components.
While traversing the border of the   rough disc in line~\ref{ln:sweeping:counterclockwise-traversal}, the robot iterates over each cell $c$ of the gap in all executions of the `for' loop in line~\ref{ln:sweeping:for}.
Since the distance between the current location $c'\in \D$ and $c$ is at most $\distPC$, the robot can correctly execute instructions in line~\ref{ln:sweeping:pick-from-c}.
In line~\ref{ln:sweeping:go-to-drop}, the robot finds the first empty cell in the direction $d$ away from $\D$.
Then, in line~\ref{ln:sweeping:drop-and-return}, the robot drops the brick at this cell and returns to $c'$.
We argue that the latter is feasible.
This is done by first finding the closest empty cell $e$ in the direction opposite to $d$.
This cell is at distance at most $\distPC$ from the   rough disc, due to the condition used in line~\ref{ln:sweeping:go-to-drop}.
Thus, the robot can remember the relative position of $e$ with respect to $c'$ because the distance between them is bounded.
The latter is due to the selected direction $d$. Indeed, the fact that direction $d$ is away from $\D$ implies that the prefix of the walk in line~\ref{ln:sweeping:go-to-drop} consisting of cells at distance at most $\distPC$ from $c'$ is of bounded length.

The above proves that each cell at distance $\distPC$ from the   rough disc, except the marker, is empty upon completion of the traversal in line~\ref{ln:sweeping:counterclockwise-traversal}.
Since the `for' loop in line~\ref{ln:sweeping:for} does not select cells at distance larger than $\distPC$ from the   rough disc, the   rough disc has a gap of width $\distPC$, i.e., the field satisfies condition \ref{it:structured-gap}.
Moreover, this implies that no lost component is created during the traversal in line~\ref{ln:sweeping:counterclockwise-traversal} because each brick dropped in line~\ref{ln:sweeping:drop-and-return} is placed next to a full cell (see line~\ref{ln:sweeping:go-to-drop}).

Checking if there exists a free component $\C$ in line \ref{ln:sweeping:if-component} can be done by traversing the border of the   rough disc because, by assumption, there are no lost components.
By the same argument, the robot can relocate the marker in line~\ref{ln:reposition-marker} of procedure \ref{alg:sweeping}.
This also ensures that the field satisfies condition \ref{it:structured-MPC}.
Thus, the field is structured, i.e., it satisfies condition \ref{it:strongly-structured:structured}.

By assumption, there are no lost components, and each brick dropped in line~\ref{ln:sweeping:drop-and-return} is placed next to a full cell.
Thus, no lost component is created during procedure \ref{alg:sweeping}, which implies that the field satisfies condition \ref{it:strongly-structured:no-lost}.
According to line~\ref{ln:sweeping:last-trip-to-marker} of procedure \ref{alg:sweeping}, upon its completion the robot is at the marker, which ensures condition \ref{it:strongly-structured:marker}.
Thus, the field is strongly structured as required.

Now we analyze the time needed to perform an execution of procedure \ref{alg:sweeping}.
For each cell $c$ selected in the `for' loop in line~\ref{ln:sweeping:for}, all the following steps take time $O(1)$, possibly except for the walk performed in line~\ref{ln:sweeping:go-to-drop} and the return from it in line~\ref{ln:sweeping:drop-and-return}.
The cells traversed in line \ref{ln:sweeping:go-to-drop} are called a \emph{sweep-walk}.
Thus, to estimate the execution time of procedure \ref{alg:sweeping}, we estimate the total length of all such walks that occur during the execution.

The call in line~\ref{ln:nest:init-sweeping} in algorithm \ref{alg:nest} takes time $O(s)$ because, for the initial   rough disc of size $1$, $O(1)$ cells need to be moved to ensure property~\ref{it:structured-gap} of the field, and for each such cell, the sweep-walk is of length $O(s)$.

We now consider one execution of procedure \ref{alg:sweeping} called in procedure \ref{alg:extend-pseudo-square}.
In order to analyze the time of executing this call, we divide sweep-walks depending on their length.
Let $W_2$ be all sweep-walks of length at most $2$ and let $W_{*}$ be all the remaining sweep-walks.
A sweep-walk $X$ is traversed by the robot because condition \ref{it:structured-gap} is violated, i.e., there is a full cell $c$ at distance less than $\distPC$ from the   rough disc.
We consider two cases.

In the first case, some brick has been dropped at $c$ in the iteration of the `while' loop of algorithm \ref{alg:nest} in which the considered call to procedure \ref{alg:sweeping} occurs.
The event of dropping the brick occurred after the previous call to procedure \ref{alg:sweeping} because this previous call, as we have proved above, left the field structured.
Thus, this occurred during an execution of procedure \ref{alg:find-next-brick} or procedure \ref{alg:return-to-marker}.
The cell $c$ does not become full as a result of a switch because, by Lemma~\ref{lem:return-nothing-lost}, all switch-components are reconnected with the component in which the agent is performing a return switch-traversal in procedure \ref{alg:return-to-marker}.
Thus, $c$ becomes full during shifting.
Let $S$ be the last segment traversed by the robot during the execution of procedure \ref{alg:find-next-brick}.
Note that the brick at $c$ has been moved from a cell $c'$ of $S$ during shifting.
The cell $c$ is at distance $2$ from $c'$.
If the walk $X$ is perpendicular to $S$, then it has 2 cells because the neighbor of $c$ in $S$ is empty by the definition of shifting.
If the walk $X$ is parallel to $S$, then $c$ is the single cell, among those filled during shifting, that results in a sweep-walk, due to the direction away from $\D$ in procedure \ref{alg:sweeping}.
Thus, in this first case, we have $O(s)$ sweep-walks in $W_2$ and at most one sweep-walk in $W_*$.

In the second case, the cell $c$ has been full in the iteration of the `while' loop of algorithm \ref{alg:nest} in which the considered call to procedure \ref{alg:sweeping} occurs.
Thus, this cell results in a sweep-walk because the robot extended the   rough disc.
However, for each such extension, there are $O(1)$ such cells $c$.
Hence, in the second case we have $O(1)$ sweep-walks in $W_2$ and $O(1)$ sweep-walks in $W_*$.

Note that the length of a sweep-walk in $W_*$ is $O(\sqrt{z}+s)=O(s)$ because the span of the final   rough disc is $O(\sqrt{z})$.
Thus, in each iteration of the `while' loop of algorithm \ref{alg:nest}, the total length of all sweep-walks is $O(s)$.
\end{proof}

Whenever procedure \ref{alg:extend-pseudo-square} is called in our algorithm, the robot is at the marker. The following lemma says that it can correctly execute the procedure.

\begin{lemma} \label{lem:extend-pseudo-square-correct}
If the robot is at the marker, carries a brick, and there are no lost components, then the robot can correctly execute procedure \ref{alg:extend-pseudo-square} in time $O(s)$.
\end{lemma}
\begin{proof}
In order to execute the instruction in line~\ref{ln:pseudo-square:extend} of procedure \ref{alg:extend-pseudo-square}, it is enough for the robot to remember whether the current rough disc is a disc.
Adding a brick to the   rough disc requires performing one traversal along its border to find the correct location for the brick.
This takes time $O(s)$ because the size of the border of a   rough disc is $O(\sqrt{z})$ which is $O(s)$.
By Lemma~\ref{lem:sweeping}, sweeping in line~\ref{ln:pseudo-square:sweeping} can be correctly executed because the robot is at the marker when it is called and,  by assumption, there are no lost components.
Moreover, the time needed to execute procedure \ref{alg:sweeping} is $O(s)$.
\end{proof}

We are now able to prove our main result, i.e.,  Theorem~\ref{thm:main}.

\begin{proof}[Proof of Theorem~\ref{thm:main}]
We argue by induction on the number of iterations of the `while' loop in algorithm \ref{alg:nest} that at the beginning of the $i$-th iteration, for $i \geq 1$, the field is strongly structured.
The claim holds for $i=1$ because, by Lemma~\ref{lem:sweeping}, the call to procedure \ref{alg:sweeping} in line~\ref{ln:nest:init-sweeping} results in a field that is structured.
Assuming that the assertion holds for some $i\geq 1$, consider the $(i+1)$-th iteration.
By Lemma~\ref{lem:return-nothing-lost}, after executing procedures \ref{alg:find-next-brick} and \ref{alg:return-to-marker} in the $(i+1)$-th iteration, there are no lost components, and the robot is at the marker.
We argue that the robot is carrying a brick.
If procedure \ref{alg:find-next-brick} did not perform shifting then (see line~\ref{ln:find-next-brick-if-brick}) the robot is carrying a brick at the end of its execution due to Lemma~\ref{lem:find-brick-correct}\ref{it:find-robot}.
If shifting has been performed, then the robot is carrying a brick at the end of shifting, which is the end of the execution of procedure~\ref{alg:find-next-brick} (see line~\ref{ln:find-next-brick-if-shifting}), due to Lemma~\ref{lem:shifting}\ref{it:shifting:penultimate}.
Note that if the robot performing the return switch-traversal in procedure \ref{alg:return-to-marker} is carrying a brick prior to a switch, then it is carrying a brick upon completion of the switch.
Thus, by an inductive argument we can state that the fact that the robot is carrying a brick at the beginning of the return switch-traversal implies that the robot is carrying a brick at the end of it.
Thus in particular, the robot is carrying a brick when procedure \ref{alg:extend-pseudo-square} is called in line~\ref{ln:nest:extend-pseudo-square} of algorithm \ref{alg:nest}.
This implies that the assumptions of Lemma~\ref{lem:extend-pseudo-square-correct} are satisfied.
By this lemma, the execution of procedure \ref{alg:extend-pseudo-square} extends the   rough disc by one brick.
The call to procedure \ref{alg:sweeping} in line~\ref{ln:pseudo-square:sweeping} of procedure \ref{alg:extend-pseudo-square} results in a field that is strongly structured, in view of Lemma \ref{lem:sweeping}.
This completes the inductive argument for the $(i+1)$-th iteration. Thus we showed that, at the beginning of each iteration of the `while' loop in algorithm \ref{alg:nest}  the field is strongly structured.

The above implies that there will be $z-2$ iterations of the `while' loop in algorithm \ref{alg:nest} and, upon the exit from the loop, there are only two components: the   rough disc of size $z-1$ and the marker.
By Proposition~\ref{pro:nest}, adding the brick from the marker in line~\ref{ln:nest:final-brick} to the rough disc results in  a nest.

By the definition of a search walk, the length of any search walk is $O(s)$.
The time needed to execute the two procedures that traverse the search walk and backtrack on it is linear in the length of the walk.
Using Lemmas~\ref{lem:sweeping} and \ref{lem:extend-pseudo-square-correct}, the overall time of algorithm \ref{alg:nest} is $O(sz)$.
By Proposition \ref{pro:lb}, this time is worst-case optimal.
\end{proof}

\section{Conclusion}

We designed a finite deterministic automaton that builds the most compact structure starting from any connected field of bricks, and does it in optimal time. 
An interesting problem yielded by our research is to characterize the classes of target structures that can be built by a single automaton, starting from any connected field of bricks in the grid.
Another problem is that of how the building task parallelizes, i.e., how much time many automata use to build some structure.

\pagebreak




\begin{thebibliography}{12}

\bibitem{AH}
S. Albers and M. R. Henzinger.
\newblock Exploring unknown environments.
\newblock SIAM J. Computing 29: 1164-1188, 2000.




\bibitem{ABRS}
B. Awerbuch, M. Betke, R. Rivest and M. Singh.
\newblock Piecemeal graph learning by a mobile robot.
\newblock Information and Computation 152(2): 155-172, 1999.




\bibitem{BFRSV}
M. Bender, A. Fernandez, D. Ron, A. Sahai and S. Vadhan.
\newblock The power of a pebble: Exploring and mapping directed graphs.
\newblock Information and Computation 176(1): 1-21, 2002.

\bibitem{BS}
M. Bender and D. Slonim.
\newblock The power of team exploration: Two robots can learn
unlabeled directed graphs.
\newblock In 35th Ann. Symp. on Foundations of Computer Science (FOCS),
pages 75-85, 1994.



\bibitem{BK78}
M. Blum and D. Kozen. 
\newblock On the power of the compass (or, why mazes are easier to search 
than graphs).  
\newblock In 19th Symposium on Foundations of Computer Science (FOCS), 
pages 132-142, 1978.

\bibitem{BS77}
M. Blum and W. Sakoda.
\newblock On the capability of finite automata in 2 and 3 dimensional space.
\newblock In 18th Ann. Symp. on Foundations of Computer Science (FOCS),
pages 147-161, 1977.

\bibitem{BRS2}
M. Betke, R. Rivest and M. Singh. 
\newblock Piecemeal learning of an unknown environment.
\newblock Machine Learning 18: 231-254, 1995.




\bibitem{BUW}
S. Brandt, J. Uitto and  R. Wattenhofer.
\newblock
A Tight Lower Bound for Semi-Synchronous Collaborative Grid Exploration. 
\newblock In 32nd International Symposium on Distributed Computing (DISC), pages 13:1-13:17,2018.

\bibitem{B78}
L. Budach.
\newblock Automata and labyrinths.
\newblock Math. Nachrichten 86: 195-282, 1978.




\bibitem{CDK}
J. Chalopin, S. Das, A. Kosowski,
Constructing a map of an anonymous graph: Applications of universal sequences,
In 14th International Conference on Principles of Distributed Systems (OPODIS), pages 119-134, 2010.



\bibitem{CELU}
L. Cohen, Y. Emek, O. Louidor and J. Uitto.
\newblock Exploring an Infinite Space with Finite Memory Scouts. 
\newblock In 28th Ann. ACM-SIAM Symp. on Discrete Algorithms (SODA), pages  207-224, 2017.

\bibitem{DFSY}
S. Das, P. Flocchini, N. Santoro and M. Yamashita.
\newblock On the computational power of oblivious robots: forming a series of geometric patterns.
\newblock In 29th Ann. ACM  Symposium on Principles of Distributed Computing (PODC), pages 67-76, 2010.

\bibitem{DP}
X. Deng and C. H. Papadimitriou.
\newblock Exploring an unknown graph.
\newblock J. Graph Theory 32(3): 265-297, 1999.

\bibitem{DFKP}  
K. Diks, P. Fraigniaud, E. Kranakis, and A. Pelc. 
\newblock Tree Exploration with Little Memory.
\newblock J. Algorithms 51(1): 38-63, 2004.


\bibitem{DPV}
Y. Dieudonn\'{e}, F. Petit and V. Villain.
\newblock Leader election problem versus pattern formation problem.
\newblock In 24th International Symposium on Distributed Computing (DISC), pages 267-281, 2010. 

\bibitem{DJMW}
G. Dudek, M. Jenkins, E. Milios, and D. Wilkes. 
\newblock Robotic Exploration as Graph Construction.
\newblock IEEE Transaction on Robotics and Automation 7(6): 859-865, 1991. 

\bibitem{DKK}
C. Duncan, S. Kobourov and V. Kumar.
\newblock Optimal constrained graph exploration.
\newblock In 12th Ann. ACM-SIAM Symp. on Discrete Algorithms (SODA), pages
807-814, 2001.

\bibitem{ELSUW}
Y. Emek, T. Langner, D. Stolz, J. Uitto, and R. Wattenhofer,  
\newblock How many ants does it take to find the food?
\newblock Theor. Comput. Sci. 608: 255-267, 2015.


\bibitem{FIPPP}
P. Fraigniaud, D. Ilcinkas, G. Peer, A. Pelc, D. Peleg, 
\newblock Graph exploration by a finite automaton, 
\newblock In 29th International Symposium on Mathematical Foundations of Computer Science (MFCS), 
LNCS 3153, 451-462, 2004.


\bibitem{FI04}
P. Fraigniaud,  and D. Ilcinkas.
\newblock Directed Graphs Exploration with Little Memory.
\newblock In 21st Symposium on Theoretical Aspects of
Computer Science (STACS), LNCS 1996, pages 246-257, 2004.

\bibitem{H89}
A. Hemmerling.
\newblock Labyrinth Problems: Labyrinth-Searching Abilities of Automata.
\newblock Volume 114 of Teubner-Texte zur Mathematik. B. G. Teubner
Verlagsgesellschaft, Leipzig, 1989.

\bibitem{H81}
F. Hoffmann.
\newblock One pebble does not suffice to search plane labyrinths.
\newblock In Fund. Computat. Theory (FCT), LNCS 117, pages 433-444, 1981.




%
%
 

\bibitem{K79}
D. Kozen.
\newblock Automata and planar graphs.
\newblock In Fund. Computat. Theory (FCT), pages 243-254, 1979. 






\bibitem{LKUW}
T. Langner,  B. Keller, J. Uitto, R. Wattenhofer,
Overcoming Obstacles with Ants,
Proc. 19th International Conference on Principles of Distributed Systems (OPODIS 2015), 1-17.

\bibitem{PaPe}
P. Panaite and A. Pelc,
\newblock Exploring unknown undirected graphs,
\newblock J. Algorithms 33(2): 81-295, 1999.


\bibitem{RKSI}
N. Rao, S. Kareti, W. Shi, and S. Iyengar.
\newblock Robot navigation in unknown terrains: Introductory survey of
non-heuristic algorithms. 
\newblock Tech. Report ORNL/TM-12410, 
Oak Ridge National Lab., 1993. 

\bibitem{R79-80}
H.A. Rollik.
\newblock Automaten in planaren Graphen.
\newblock Acta Informatica 13: 287-298, 1980.
\bibitem{S51}
C. E. Shannon.
\newblock Presentation of a maze-solving machine.
\newblock In 8th Conf. of the Josiah Macy Jr. Found. (Cybernetics),
pages 173-180, 1951.

\bibitem{SY}
I. Suzuki and M. Yamashita.
\newblock Distributed anonymous mobile robots: formation of geometric patterns.
\newblock SIAM J. Computing 28(4): 1347-1363, 1999.

\end{thebibliography}
\end{document}